\documentclass[12pt]{amsart}
\usepackage[T1]{fontenc}
\usepackage{kpfonts}
\usepackage[dvipsnames]{xcolor}
\usepackage{pdfpages}
\usepackage{sgame}
\usepackage{accents}
\usepackage{tikz-cd}
\usepackage{float}
\usepackage[round]{natbib}
\usepackage{multirow}
\usepackage{dutchcal}
\usepackage{pdfpages}
\usepackage{bbm}
\usepackage{sgame}
\usepackage{mathtools}
\usepackage{amsthm}
\usepackage{enumitem}
\usepackage[margin=1.25in]{geometry}
\usepackage{caption}
\usepackage{subcaption}
\usepackage{epigraph}
\usepackage{listings}
\usetikzlibrary{calc}
\usetikzlibrary{shapes,arrows}
\usepackage{tcolorbox}

\DeclareMathOperator\inter{int}

\newtheorem{theorem}{Theorem}[section]
\newtheorem{proposition}[theorem]{Proposition}
\newtheorem{lemma}[theorem]{Lemma}
\newtheorem{corollary}[theorem]{Corollary}
\newtheorem{claim}[theorem]{Claim}

\theoremstyle{definition}

\newtheorem{definition}{Definition}

\newtheorem{remark}[theorem]{Remark}

\usepackage{hyperref}
\hypersetup{
    colorlinks=true,
    linkcolor=CadetBlue,
    filecolor=CadetBlue,      
    urlcolor=CadetBlue,
    citecolor=Mulberry,
}

\usepackage{setspace}

\definecolor{backcolour}{rgb}{0.63, 0.79, 0.95}

\lstdefinestyle{mystyle}{
  backgroundcolor=\color{backcolour},
  basicstyle=\ttfamily\footnotesize,
  breakatwhitespace=false,         
  breaklines=true,                 
  captionpos=b,                    
  keepspaces=true,                 
  numbers=left,                    
  numbersep=5pt,                  
  showspaces=false,                
  showstringspaces=false,
  showtabs=false,                  
  tabsize=2
}

\providecommand{\keywords}[1]{\textbf{\textit{Keywords:}} #1}
\providecommand{\jel}[1]{\textbf{\textit{JEL Classifications:}} #1}


\begin{document}
\author{Marilyn Pease \and Mark Whitmeyer}
\thanks{MP: Kelley School of Business, Indiana University, \href{marpease@iu.edu}{marpease@iu.edu} \& MW: Arizona State University, \href{mailto:mark.whitmeyer@gmail.com}{mark.whitmeyer@gmail.com}. We thank Arjada Bardhi, Henrique Castro-Pires, Chris Chambers, Duarte Gon\c{c}alves, Faruk Gul, Ilia Krasikov, Pietro Ortoleva, and Sungmin Park for their comments. We are especially grateful to Xiao Lin, who contributed significantly to the material in \S\ref{lotteries}; and to Jian Li, who was an exceptionally helpful discussant. Seminar and conference audiences at ASU, Edinburgh, Georgetown, UGA, OSU, UCL, UWO, the 2024 SEA meetings, ESWC 2025, and the 2025 Women in Economic Theory conference provided useful feedback. Konstantin von Beringe and Titus Gao provided excellent research assistance. This paper was previously circulated under the title "How to Make an Action Better." Draft date: \today.}
\title{How to Make an Action Attractive}

\begin{abstract}
A policymaker often wants to steer a decision-maker toward one of two actions, but lacks reliable knowledge of how the decision-maker perceives uncertainty or evaluates risk. We formalize a notion of robust paternalism: a modification \(\hat{a}\) of a desired action \(a\) is robustly more attractive than \(a\) relative to \(b\) if, for every belief over states and every increasing concave utility function, whenever the decision-maker prefers \(a\) to \(b\), she also prefers \(\hat{a}\) to \(b\). We characterize all such modifications directly in terms of state-dependent payoffs and discuss applications to political competition, bilateral trade, insurance, and information acquisition.
\end{abstract}

\maketitle
\\
\jel{D00; D80}

\section{Introduction}

A central goal of paternalistic policy is to make a desired behavior more appealing than a tempting alternative: think of vaccination versus non-vaccination, retirement saving versus immediate consumption, safe versus risky driving, or smoking versus non-smoking. In practice, however, the policymaker rarely knows how the population of individuals perceives the underlying uncertainty or how sensitive they are to risk. A policy that looks like an ``obvious improvement'' under one belief or one utility can backfire under another.

This paper develops a notion of \emph{robust paternalism} for a simple but widely applicable environment. There are two actions facing an agent in a decision problem, \(a\) and \(b\), each delivering a state-dependent monetary payoff. A policymaker can modify the recommended action \(a\) into \(\hat a\) (via a subsidy, a transfer scheme, a contractual redesign, or any intervention that changes the payoff profile across states), but does not observe the decision-maker's belief over states and does not know her von Neumann-Morgenstern utility beyond monotonicity and (in our main result) concavity. We ask: what properties of \(\hat a\) guarantee that \(\hat a\) is \emph{never less attractive than \(a\) versus \(b\)}?  

The same basic question appears in strategic and institutional settings in which a designer must choose what to offer a heterogeneous population, such as product design and electoral competition. A firm or political party may have the goal of replacing \(a\) with \(\hat a\) and not shrinking the set of people who prefer it to \(b\), regardless of their beliefs or risk preferences. Such a goal is natural when the designer has little reliable knowledge of this heterogeneity. 

Consider, for instance, a consumer deciding between two firms' products. Firm 2 produces product \(b\), while firm 1 can choose to produce either product \(a\) or product \(\hat a\). If the consumer prefers \(a\) to \(b\), what are the properties of \(\hat a\) that guarantee that the consumer also prefers it to \(b\)? In the political venue, suppose party \(A\) must choose between candidates \(a\) and \(\hat a\) to run against party \(B\)'s candidate \(b\). When would party \(A\) benefit more from running candidate \(\hat a\) over candidate \(a\), given that \(b\) is the opponent? Allowing for a population with wildly heterogeneous beliefs and utility functions, what are the characteristics of \(\hat a\) that would ensure at least as many people choose it over \(b\) as they would \(a\)?

The knee-jerk response to these questions is the following guess: when the decision-maker's (DM's) utility is known only to be within the class of increasing-in-money utility functions, it must be that \(\hat{a}\) corresponds to a first-order stochastic dominance improvement over \(a\). Likewise, when the DM is known also to be risk averse, it must be that \(\hat{a}\) corresponds to a second-order stochastic dominance improvement over \(a\). However, these answers do not stand up to scrutiny--in particular, we do not know the lotteries produced by the DM's subjective belief, so we cannot speak directly of dominance of lotteries.\footnote{In fact, as we reveal in \S\ref{lotteries}, even if we know the DM's belief, a dominance improvement of the lottery induced by \(a\) is sufficient, but not necessary--a weaker condition suffices.} State-wise dominance of \(a\) or \(b\) is also a promising answer, but it turns out that this is too strong when the DM is known to be risk averse.

In our two main results, Theorems \ref{bigone} and \ref{theorem38}, we fully characterize these relations in terms of the primitives of the environment--the state-dependent payoffs of actions \(a\), \(\hat{a}\), and \(b\). Theorem \ref{bigone} concerns the case of a risk-averse DM. We show that the DM preferring \(a\) to \(b\) must imply that she also prefers \(\hat{a}\) to \(b\) if and only if \(\hat{a}\) dominates an action whose payoffs are a convex combination of the payoffs to \(a\) and \(b\). Intuitively, if $a$ is preferred to $b$, then mixing some of $a$ with $b$ is better than $b$ as well. Then, making $\hat{a}$ dominate such a mixture can only further increase its attractiveness versus $b$. 

If we do not assume the DM to be risk averse, our other main result (Theorem \ref{theorem38}) states that the necessary and sufficient conditions become much more restrictive: \(\hat{a}\) must dominate \(a\) or \(b\) (in a state-wise sense). If $u$ is concave, we know from Theorem \ref{bigone} that significant structure on payoff transformations is required. Convexity also imposes structure, but in the opposite manner. Combining these requirements leaves only dominance as necessary (and it is obviously sufficient).


With these results in hand, we discuss several examples. We begin in \S\ref{subsec:smoking} with the robust-paternalism motivation: a policymaker designs budget-balanced transfers to increase smoking cessation. We then turn to elections and contracting--bilateral trade of a risky asset and design of an insurance policy. In each, the economic message is that a robust improvement \(\hat a\) must ``move \(a\) toward \(b\),'' becoming more (less) risky than \(a\) precisely when \(b\) is more (less) risky. Somewhat counterintuitively, this means that when \(b\) is the riskier alternative, a robustly more attractive \(\hat a\) cannot be obtained by simply ``insuring'' \(a\) or smoothing its payoffs: the way to expand appeal across heterogeneous beliefs and risk attitudes is to also tilt \(a\) in the direction of \(b\)'s risk; \textit{viz.}, make \(a\) riskier. \textit{To robustly make action \(a\) appeal more to the risk averse, the designer makes it riskier.}

We finish the paper by studying four natural variants of our main question. Theorem \ref{bigone} assumes two dimensions of uncertainty about the protagonist; both her belief and utility function. In \S\ref{lotteries}, we shut one of these down, positing that the DM's belief is known, then characterize, in Theorem \ref{theoremlottery}, ``how to make a lottery better.'' \S\ref{many} restores uncertainty about beliefs; there, we ask for robust improvements to an action vis-a-vis not just one but multiple alternatives. In Proposition \ref{manyactions}, we show that in a broad class of decision problems, a straightforward aggregation of Theorem \ref{bigone}'s conditions is necessary and sufficient for one action to improve against many. 

Our final section endogenizes the DM's belief. Namely, \S\ref{infoac} explores information acquisition, in which our question becomes: how must \(a\) relate to \(\hat{a}\) in a way that guarantees that \(\hat{a}\) is chosen more frequently in the DM's problem of binary-choice with flexible information acquisition? With two states, we show that \(\hat{a}\) dominating \(a\) or \(b\) is necessary and sufficient. Notably, we illustrate that with three or more states, dominating \(a\) is no longer sufficient. That is, \(\hat{a}\) can be made better than \(a\) in a strong sense (state-wise dominance), \textit{yet be chosen less}.

\subsection{Related Work}

The body of work studying decision-making under uncertainty is sizeable. The work closest to this one is \cite{safety}. There, we formulate a binary relation between actions: action \(a\) is ``safer" than \(b\) if the set of beliefs at which \(a\) is preferred to \(b\) grows larger, in a set inclusion sense, when we make the DM more risk averse. Both \cite{markcs} and \cite{markscs} built off of this paper: the former revisits \cite{rabin2000risk}'s calibration exercise in a subjective world, and the latter studies robust comparative statics.

\cite{ROTHSCHILD1970225} is a seminal work that characterizes transformations of lotteries that are preferred by all risk-averse agents. \cite{aumann2008economic} formulate a ``measure of riskiness'' of gambles, as do \cite{foster2009operational} (who are subsequently followed up by \cite{bali2011generalized} and \cite{riedel2015foster}). Crucially, these indices and measures correspond to inherently stochastic objects--the lotteries at hand. Our conception of an improvement to an action concerns comparisons of state-dependent payoffs, which are themselves non-random objects (they are just real numbers).

Naturally this work is also connected to the broader literature studying actions that are comparatively friendly toward risk. In addition to the aforementioned paper of ours, \cite{safety}, which, like this one, centers around a decision-maker's state-dependent payoffs to actions, this literature includes \cite{hammond1974simplifying}, \cite{lambert1979attitudes}, \cite{karlin1963generalized}, \cite{jewitt1987risk}, and \cite{jewitt1989choosing}.
Notably, they are statements about lotteries, \textit{viz.}, random objects.

Our research question can be reformulated as a revealed-preference exercise. This equivalence is central to deriving our result in the known-beliefs portion of the paper (\S\ref{lotteries}), where we use a remarkable result of \cite{fishburn1975separation} ``off the shelf.'' Existing results are not as useful in our main environment--with uncertain beliefs \textit{and} utilities. In a similar environment to our known-beliefs setting, \cite{chompchomp} asks for the properties of lotteries \(L_1\), \(L_2\), \(\hat{L}_1\) and \(\hat{L}_2\), equivalent to a DM preferring \(\hat{L}_1\) to \(\hat{L}_2\) whenever she prefers \(L_1\) to \(L_2\). To elaborate, \cite{chompchomp} looks at two arbitrary binary menus of lotteries, whereas in \S\ref{lotteries}, we look at two finite (not-necessarily-binary) menus of lotteries in which all but one of the lotteries are preserved across menus. Our main environment, in which we require robustness over beliefs as well, has no parallel in \cite{chompchomp}.

A number of other papers also bear mention. \cite{gilboa2022were} ask when beliefs can justify a collection of prize and certainty-equivalent pairs.\footnote{Their main result is similar to the result we use from \cite{fishburn1975separation}. This is due to the fact that both papers' questions boil down to whether particular separating hyperplanes exist, which is a common theme running through the revealed-preference literature.} \cite{richter1978revelations} scrutinize what a collection of binary choices between actions reveals about an agent's beliefs. \cite{echenique2015savage} characterize the market behavior of a risk-averse subjective expected utility maximizing agent, showing that consistent behavior is equivalent to satisfying the ``strong axiom of revealed subjective utility.'' 

This paper also inhabits the comparative statics realm--we change an aspect of a decision problem and see how it affects a decision-maker's choice. Our robustness criterion as well as the simplicity of our setting distinguishes our work from the standard pieces, e.g., \cite{milgrom1994monotone}, \cite{edlin1998strict}, and \cite{athey2002monotone}. The works involving aggregation (\cite{quah2012aggregating}, \cite{choi2017ordinal}, and \cite{kartik2023single}) are closer still--as this inherently corresponds to distributional robustness--but none leave as free parameters both the distribution over states and the DM's utility function, as we do. Special mention is due to \cite{curello2019preference}, who conduct a robust comparative statics exercise in which an analyst, with limited knowledge of an agent's preferences, predicts the agent's choice across menus.

In our final section, we explore properties of the new action, \(\hat{a}\), that make it more likely to be selected than \(a\) when the alternative is \(b\) if information is endogenously acquired. This property is similar in spirit to the observation of \cite{matvejka2015rational} that new actions added to a menu may ``activate'' previously unchosen actions. \cite{muller2021rational} provide a full characterization of this phenomenon. One crucial distinction between our analysis and theirs is that they explore \textit{additions} whereas our modification is a \textit{replacement}.

\section{Model}

There is a topological space of states, \(\Theta\), which is endowed with the Borel \(\sigma\)-algebra, and which we assume to be compact and metrizable. \(\theta\) denotes a generic element of \(\Theta\). We denote the set of all Borel probability measures on \(\Theta\) by \(\Delta \equiv \Delta \left(\Theta\right)\). There is also a decision-maker (DM), who is endowed with two actions, \(a\) and \(b\). \(A = \left\{a,b\right\}\) denotes the set of actions, and each action \(\tilde{a} \in A\) is a bounded Borel-measurable function from the state space to the set of outcomes, \(\tilde{a} \colon \Theta \to \mathbb{R}\). For convenience, for any \(\tilde{a} \in A\), we write \(\tilde{a}_\theta \equiv \tilde{a}\left(\theta\right)\). Given a probability distribution over states \(\mu \in \Delta\), an action is a (simple) lottery. 

We further specify that no action \(\tilde{a} \in A\) is weakly dominated by the other, and partition \(\Theta\) into three sets:
\[\mathcal{A} \coloneqq \left\{\theta \in \Theta \colon a_{\theta} > b_{\theta}\right\}, \quad \mathcal{B} \coloneqq \left\{\theta \in \Theta \colon a_{\theta} < b_{\theta}\right\}, \quad \text{and} \quad \mathcal{C} \coloneqq \left\{\theta \in \Theta \colon a_{\theta} = b_{\theta}\right\}\text{.}\]
The DM is an expected-utility maximizer, with a von Neumann-Morgenstern utility function defined on the outcome space \(u \colon \mathbb{R} \to \mathbb{R}\). We assume that the DM's utility lies in some convex class of continuous strictly increasing utility functions \(\mathcal{U}\). In our results, we specialize to two particular such classes of interest; \(\mathcal{U}_{\uparrow}\), in which no further restrictions are placed on the utilities; and \(\mathcal{U}_{RA}\), in which utilities are also assumed to be weakly concave. Our main focus is the latter case.


Given \(A\), we are interested in how a modification of \(a\) affects the DM's choice of action. That is, we will modify \(a\) to some new \(\hat{a}\), in which case the new menu is \(\left(\hat{a},b\right)\), and examine when \(\hat{a}\) must be (strictly) preferred to \(b\) whenever \(a\) is. Formally,
\begin{definition}\label{definition1}
 Action \(\hat{a}\) is \textit{\(b\)-superior} to action \(a\) over \(\Delta\) and \(\mathcal{U}\) if for all \(\mu \in \Delta\), for all \(u \in \mathcal{U}\),
 \[\mathbb{E}_{\mu}u\left(a_{\theta}\right) \underset{(>)}{\geq} \mathbb{E}_{\mu}u\left(b_{\theta}\right) \quad \Rightarrow \quad \mathbb{E}_{\mu} u\left(\hat{a}_{\theta}\right) \underset{(>)}{\geq}  \mathbb{E}_{\mu} u\left(b_{\theta}\right)\text{.}\]\end{definition}

\section{Robust Improvements}\label{s:RobustImprovements} 

We now present our main result, characterizing the conditions under which $\hat{a}$ is \(b\)-superior to $a$ when \(\mathcal{U} = \mathcal{U}_{RA}\). To do so, we first define a \textit{mixture} of actions \(a\) and \(b\) to be the action \(a^\lambda\) that yields payoff \[a^{\lambda}_\theta \coloneqq \lambda a_{\theta} + \left(1-\lambda\right) b_\theta\]
in each state \(\theta \in \Theta\) for some \(\lambda \in \left[0,1\right]\).

Our main result is, understanding dominance in a state-wise sense,\footnote{Specifically, \(\tilde{a}\) dominates \(a\) if \(\tilde{a}_\theta \geq a_\theta\) for all \(\theta \in \Theta\), permitting equality in all states.}
\begin{theorem}\label{bigone}
    Fix \(a\) and \(b\). Action \(\hat{a}\) is \(b\)-superior to action \(a\) over \(\Delta\) and \(\mathcal{U}_{RA}\) if and only if \(\hat{a}_\theta > b_\theta\) for all \(\theta \in \mathcal{A}\) and \(\hat{a}\) (weakly) dominates a mixture of \(a\) and \(b\). 
\end{theorem}
\begin{proof}
    \(\left(\Leftarrow\right)\) Suppose \(\hat{a}\) is a mixture of \(a\) and \(b\) (as subsequent dominance only makes \(\hat{a}\) more enticing). As \(\hat{a}_\theta > b_\theta\) for all \(\theta \in \mathcal{A}\), \(\hat{a} \neq b\). Let \(L_a\) be the lottery that pays out \(a_\theta\) with probability \(\mu\left(\theta\right)\); \(L_b\), the lottery that pays out \(b_\theta\) with probability \(\mu\left(\theta\right)\); and \(L_{\hat{a}}\), the lottery that pays out \(\hat{a}_\theta\) with probability \(\mu\left(\theta\right)\). Also suppose \(\mu \in \Delta\) is such that the DM prefers \(a\) to \(b\), i.e., \(L_a \succeq L_b\) (the proof for the case where we posit \(L_a \succ L_b\) follows analogously, so we omit it). 
    
    By independence (which is implied by expected utility), for any \(\lambda \in \left[0,1\right]\),
    \[L^\lambda \coloneqq \lambda L_a + \left(1-\lambda\right) L_b \succeq L_b\text{.}\]
    Moreover, by assumption there is some \(\lambda^* \in \left[0,1\right]\) for which \(\lambda^* a_\theta + \left(1-\lambda^*\right) b_\theta = \hat{a}_\theta\) for all \(\theta \in \Theta\). This implies that \(L^{\lambda^*}\) is a mean-preserving spread of \(L_{\hat{a}}\). Consequently, as the DM is risk averse, \(L_{\hat{a}} \succeq L^{\lambda^*} \succeq L_b\).

    \(\left(\Rightarrow\right)\) See Appendix \ref{Thm_Main}.
    \end{proof}

    The sufficiency proof can be summarized as follows. Let $\hat{a}$ be a mixture of $a$ and $b$, and consider a belief $\mu$ such that $a$ is preferred to $b$. Then, if $a$ is preferred to $b$ at $\mu$, any average of the lotteries corresponding to \(a\) and \(b\) at belief \(\mu\) must also be preferred to \(b\), as expected-utility DMs subscribe to the independence axiom. Moreover, as the DM is risk averse, any such randomization over lotteries is inferior to the lottery obtained by ``removing the extra randomness'' by giving the DM the state-dependent expected value instead. This is precisely a mixture of \(a\) and \(b\) as we define it.


    The necessity proof of Theorem \ref{bigone} is a little more involved, and has been relegated to the appendix. It can be summarized as follows. First, if $\hat{a}$ is $b$-superior, then it must be that $\hat{a}_\theta > b_\theta$ for all $\theta \in \mathcal{A}$: in any state where the DM strictly prefers $a$, she must also strictly prefer $\hat{a}$. Second, in order to show that $b$-Superiority implies that $\hat{a}$ dominates a mixture of $a$ and $b$, we tackle the contrapositive. If such a dominance did not hold, we show that it must not hold for a two-state (\(\left\{\theta,\theta'\right\}\)) restriction of the problem, where \(a\) is strictly preferred to \(b\) in state \(\theta\) and \(b\) to \(a\) in \(\theta'\). In other words, if \(\hat{a}\) is \(b\)-superior, but \(\hat{a}\) does not dominate a mix of \(a\) and \(b\), then it does not dominate a mix of \(a\) and \(b\) only over states $\theta$ and $\theta^\prime$. From there, we argue that it must be that $\hat{a}_{\theta^\prime} < a_{\theta^\prime} < b_{\theta^\prime}$ and $\hat{a}_{\theta} > a_{\theta} > b_{\theta}$, which allows us to construct a concave, single-kinked (``hockey stick'') utility function that either makes it so the DM prefers \(a\) to \(b\) for almost all beliefs but \(\hat{a}\) to \(b\) only for some,  or \(b\) to \(\hat{a}\) for almost all beliefs but \(b\) to \(a\) only for some. Either way, this violates \(b\)-superiority.

    We note that Theorem \ref{bigone} can easily be extended to the case in which a lower-bound (in terms of risk-aversion) for the DM is known. That is, suppose there is some (weakly) risk-averse \(\bar{u}\) of which the DM's utility is known to be some strictly increasing (weakly) concave transformation: \(u = \phi \circ \bar{u}\). Then, we just redefine \(a\) from map \(\theta \mapsto a(\theta)\) to \(\theta \mapsto \bar{u}(a(\theta))\) and likewise \(b\), before applying the theorem. Similarly, as second-order stochastic dominance is preserved under independent noise,\footnote{See, for instance, the discussion on page 3 of \cite{pomatto2020stochastic}.} the theorem also holds when we ask for robustness to aggregate risk.
\subsection{Intuition}\label{twostates}
To gain intuition for Theorem \ref{bigone}, consider the two-state environment, which is easy to visualize. The state is either $0$ or $1$, and we specify without loss of generality that \(a_{0} > b_{0}\) and \(b_{1} > a_{1}\). \(\mu \in \left[0,1\right]\) denotes the DM's belief that the state is \(1\).

There is  a cutoff belief at which the DM is indifferent between actions. Let $\bar{\mu}_u$ be the belief at which the DM is indifferent between actions $a$ and $b$ with utility $u$ and $\hat{\mu}_u$ be the belief at which the DM is indifferent between $\hat{a}$ and $b$ with the same utility. To see this graphically, let 
\[\ell_a(\mu) = a_0 (1-\mu) + a_1 \mu\text{,}\]
be the expected payoff (in money) of action \(a\) as a function of the belief \(\mu \in \left[0,1\right]\). In turn, lines
\[\ell_{\hat{a}}(\mu) = \hat{a}_0 (1-\mu) + \hat{a}_1 \mu, \quad \text{and} \quad \ell_b(\mu) = b_0 (1-\mu) + b_1 \mu\]
are the expected payoffs to actions \(\hat{a}\) and \(b\), respectively. In Figure \ref{fig1}, these are the blue (solid), red (dashed), and green (dotted) lines. The top panel of the figure shows the monetary payoffs, or the risk-neutral utility function, while the bottom panel depicts utilities with \(u\left(x\right) = \sqrt{1+x} - 3\). Then, $\hat{a}$ being $b$-superior to $a$ is equivalent to $\bar{\mu}_u  \leq \hat{\mu}_u$ for all concave $u$.

\begin{figure}
    \centering
    \includegraphics[width=.6\linewidth]{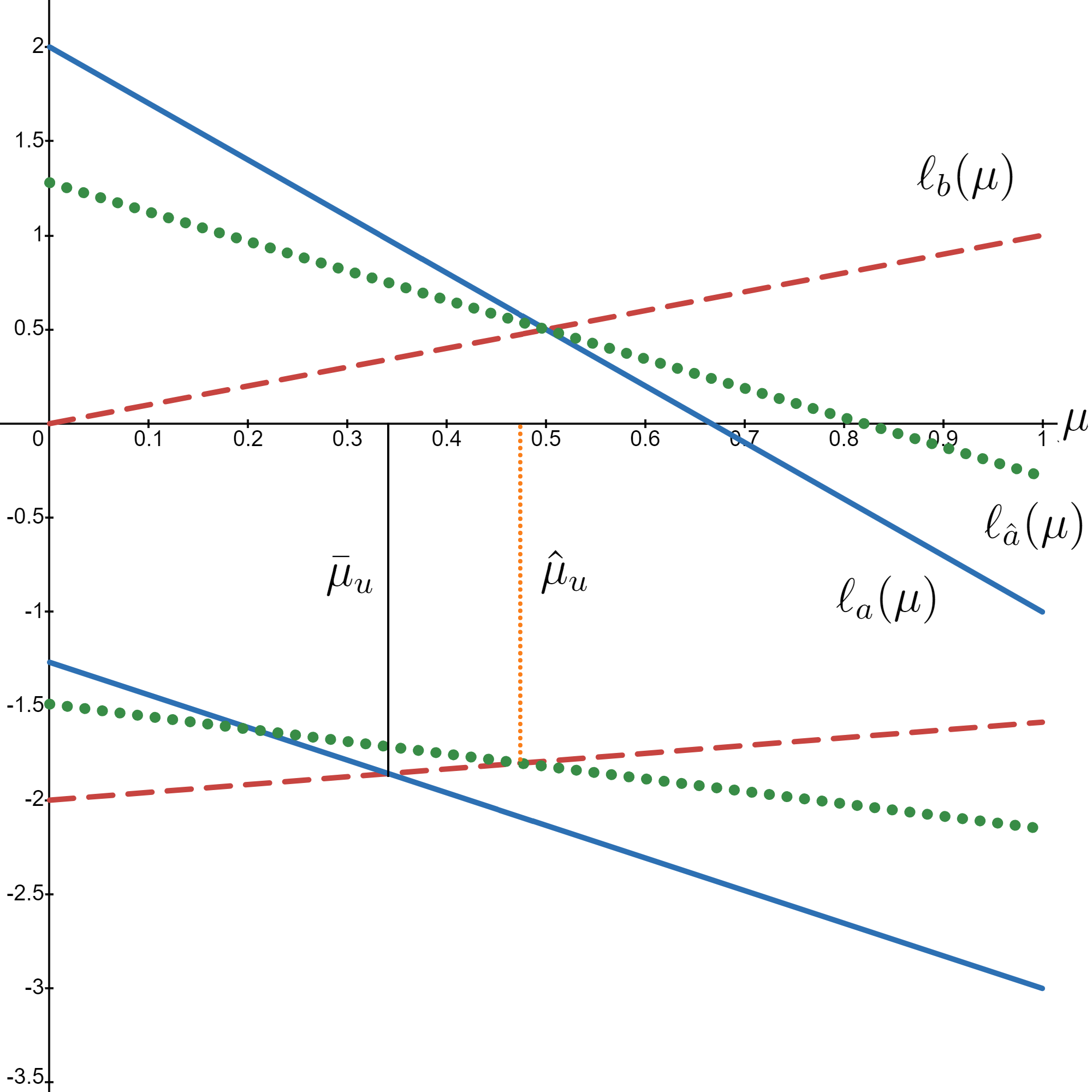}
    \caption{Improving \(a\). \href{https://www.desmos.com/calculator/rpxu7axgdb}{Try it yourself (rotate by moving the slider)!}}
    \label{fig1}
\end{figure}

Now we can think about $b$-superiority in the context of Theorem \ref{bigone}. It tells us that $\hat{a}$ must at least weakly dominate a mixture of $a$ and $b$. For now, just consider an $\hat{a}$ that is a mixture of $a$ and $b$. The corresponding line, $\ell_{\hat{a}}$, is a counterclockwise rotation of $\ell_a$. We know from the theorem that for any belief at which the DM prefers $a$ to $b$ ($\mu \leq \bar{\mu}_u$), she prefers $\hat{a}$ to $b$ as well. In other words, if $a$ was better than $b$ at $\mu$, then mixing ``some'' $a$ with $b$ is better than $b$ at $\mu$ as well. Consequently, the set of beliefs at which $\hat{a}$ is preferred must be \emph{at least} as large as the set where $a$ is preferred. Naturally, making $\hat{a}$ better than a mixture of $a$ and $b$ only strengthens its attractiveness versus \(b\).

Formally, we can write the following result. 
\begin{corollary}\label{corr1}
    \(\hat{a}\) is \(b\)-superior to \(a\) over \(\left[0,1\right]\) and \(\mathcal{U}_{RA}\) if and only if \(\hat{a}_0 > b_0\) and there exists a \(\lambda \in \left[0,1\right]\) such that \(\ell_{\hat{a}}\) lies weakly above the line \(\lambda \ell_a + \left(1-\lambda\right)\ell_b\).
\end{corollary}

\subsection{Beyond Expected Utility} 

Before going on, we wish to highlight one surprising feature of our theorem; namely, the fact that the sufficiency direction--and, therefore, the entire Theorem \ref{bigone}--is valid for far more general preferences than expected utility. We understand an action as a vector in \(\mathbb{R}^{\Theta}\) and suppose that a DM's preferences are represented by the monotone and concave \(V \colon \mathbb{R}^{\Theta} \to \mathbb{R}\), where she (strictly) prefers \(a\) to \(b\) if and only if \(V(a) \geq \ (>) \ V(b)\).

Suppose that \(\hat{a}\) dominates a mixture of \(a\) and \(b\), \(\hat{a} \neq b\), and the DM (strictly) prefers \(a\) to \(b\). Then, for some \(\lambda \in \left(0,1\right]\), by concavity and monotonicity,
\[V\left(\hat{a}\right) \geq V\left(\lambda a + (1-\lambda) b\right) \geq \lambda V(a) + (1-\lambda) V(b) \underset{(>)}{\geq} V(b) \text{.}\]
In fact, if we dropped the requirement that strict preferences are preserved, we could weaken the requirement on \(V\) to be merely quasi-concave.

\subsection{A Broader Class of Utilities}\label{ss:BroadUtil}
In the work up to this point, we have assumed mild but meaningful structure on the DM's utility function. More is better (\(u\) is strictly increasing in the reward), and the DM is risk-averse (\(u\) is, at least weakly, concave). Suppose we remove the assumption of risk aversion. We still assume that \(u\) is strictly increasing but now allow it to take arbitrary shape. Recall that \(\mathcal{U}_{\uparrow}\) denotes this class of utilities.

As we now document, without the assumption of risk aversion, the only transformations that must make \(b\) less attractive are unambiguous dominance improvements from \(a\) to \(\hat{a}\). This is intuitive; think back to Corollary \ref{corr1}. There, the necessary relationship between \(\hat{a}\) and \(a\) is \(\ell_{\hat{a}}\)'s dominance of a counterclockwise rotation of \(\ell_a\). However, the opposite modification is needed when \(u\) is convex--which is now permitted--that is, \(\ell_{\hat{a}}\) needs to dominate a \textit{clockwise} rotation of \(\ell_a\). This leaves only dominance improvements.

\begin{theorem}\label{theorem38}
    Action \(\hat{a}\) is \(b\)-superior to action \(a\) over \(\Delta\) and \(\mathcal{U}_{\uparrow}\) if and only if for all \(\theta \in \mathcal{A}\) \(\hat{a}_\theta > b_\theta\) and \(\hat{a}\) dominates \(a\) or \(b\).
\end{theorem}
\begin{proof}
    See Appendix \ref{appendix}.
\end{proof}

\section{Applications}

\subsection{Robust Paternalism}\label{subsec:smoking}

Consider the problem of a policymaker who seeks to increase the fraction of smokers who attempt to stop smoking. There are two actions: \(a\), attempt to stop smoking; and \(b\), make no such attempt. For simplicity, we assume that there are two states: \(\Theta \coloneqq\{\theta_S,\theta_F\}\), mnemonic for ``succeed'' and ''fail.'' Accordingly, we posit \(b_\theta = 0\) for all \(\theta\in\Theta\) and \(a_S>0>a_F\), so that attempting to quit is desirable in state \(S\) but undesirable in state \(F\).

The policymaker can modify the desired action by attaching success- and failure-contingent transfers to the act of attempting to quit.
Let \(t=(t_S,t_F)\in\mathbb{R}^\Theta\) denote transfers from the policymaker to the smoker, contingent on verifiable markers (e.g., biomarkers or program compliance).
The payoff of the modified quit action is then \(\hat a_\theta = a_\theta+t_\theta\) for all \(\theta\in\Theta\).

The policymaker holds a belief \(q\in\Delta(\Theta)\), whereas smokers may have heterogeneous subjective beliefs \(\mu\in\Delta(\Theta)\) and heterogeneous risk-averse utilities \(u\in \mathcal U_{\mathrm{RA}}\). \(\Pi\) denotes the population-wide probability distribution over \(\Delta(\Theta)\times \mathcal U_{\mathrm{RA}}\). The policymaker's objective is the adoption rate
\[
\alpha(t)\coloneqq \Pi\left(\left\{(\mu,u)\in\Delta(\Theta)\times \mathcal U_{\mathrm{RA}}\colon\ \mathbb{E}_\mu u(a_\theta+t_\theta)\ge u(0)\right\}\right).
\]

We contrast two potential policy objectives. In both, we impose the \textit{ex ante} budget balance constraint \(\mathbb{E}_q[t_\theta] \le 0\). The first policy objective seeks to maximize \(\alpha(t)\) only subject to budget balance:  
\[
\max_{t\in\mathbb{R}^\Theta}\ \alpha(t)
\quad \text{s.t.}\quad
\mathbb{E}_q[t_\theta]\le 0.\]
The second seeks to find the Pareto-improving policy to maximize \(\alpha(t)\). In other words, maximization is subject to the constraint that any type who would attempt quitting under \(a\) must still attempt under \(\hat a=a+t\). This is Pareto-improving in the sense that no DM pivots from action \(a\), which the policymaker knows to be better for them in the long run than action \(b\).
By Definition~\ref{definition1}, this is the requirement that \(\hat a\) is \(b\)-superior to \(a\) over \(\Delta(\Theta)\) and \(\mathcal U_{\mathrm{RA}}\):
\[\max_{t\in\mathbb{R}^\Theta}\ \alpha(t)
\quad \text{s.t.}\quad
\mathbb{E}_q[t_\theta]\le 0,
\ \text{and}\ 
\hat a \text{ is \(b\)-superior to \(a\) over \(\Delta(\Theta)\) and \(\mathcal{U}_{\mathrm{RA}}\)}.\]

In the first problem, penalizing relapse--\textit{viz.,} setting \(t_F < 0\)--can be optimal. Such a transfer shifts incentives toward the good contingency of quitting. A reward in state \(\theta_S\) financed by a penalty in state \(\theta_F\) effectively offers a bet on being in \(\theta_S\), and therefore mainly attracts smokers who are sufficiently optimistic about quitting and not too risk averse. Such schemes are most attractive when the key margin for optimizing adoption lies among these confident and not-too-risk-averse smokers.

However, as Theorem \ref{bigone} reveals, such relapse penalties are not compatible with Pareto-improving policies. Instead, the only permissible modifications are ``insurance-like'' in that they do not worsen the downside of attempting to quit. Budget balance then requires that any such downside relief be financed by reducing the upside unless outside funds are available.\footnote{To see this clearly, note that \(b\)-superiority adds three constraints in this environment: \(a_S+t_S > 0\), \(a_S+t_S \geq \lambda a_S\), and \(a_F+t_F\geq \lambda a_F\) for some \(\lambda \in [0,1]\). In particular, note that the last constraint implies that $t_F \geq 0$, and so by budget balance, $t_S\leq 0$.} 

The takeaway is therefore that any Pareto-improving policy must make the state of failure more palatable. As an example, consider ``Every Try Counts", which is a program launched by the U.S. Food and Drug Administration (FDA) in 2018. The goal of the program is to normalize failure when attempting to quit smoking and instead frame failure as progress towards the eventual goal of quitting. Viewing failure in a more positive light lessens the mental burden of being in state \(F\), so that \(t_F > 0\).

\subsection{Choosing Which Candidate to Run}

There is a unit mass of voters and two candidates \(a\) and \(b\) running for office. There is an unknown state of the world \(\theta \in \Theta\), and candidate \(a\) (\(b\)) produces monetary benefit \(a_\theta\) (\(b_\theta\)) to each voter in every state \(\theta \in \Theta\). Neither candidate dominates the other in the sense that there is at least one state of the world in which each provides a strictly higher monetary benefit to the populace than the other.

The party backing candidate \(a\) is contemplating whether to replace him with a different candidate \(\hat{a}\). Crucially, it knows neither the risk preferences of the populace (which may be heterogeneous) nor the beliefs of the populace (also possibly heterogeneous). Accordingly, it makes its choice based on a robust criterion, asking what are the candidates who have a better chance of winning versus \(b\), no matter the population's (concave) utilities or beliefs. Who cannot be worse than \(a\) versus \(b\)? Theorem \ref{bigone} tells us the answer:
\begin{remark}
    Candidate \(\hat{a}\) is no worse than \(a\) versus \(b\) if and only if she dominates a mixture of \(a\) and \(b\).
\end{remark}

In the context of politics, this can be understood, roughly, as a move to the political center. This strategy has been explicitly used in the past. For example, after Republican Ronald Reagan's landslide presidential victory in 1984, a group of Democrats founded the Democratic Leadership Council (DLC) in 1985 with the explicit goal of shifting the party towards the center in order to win elections. Likewise, Tony Blair's ``New Labour'' in the 1990s re-branded the Labour Party, focusing more on market economics and less on social redistribution in an attempt to broaden electoral appeal. More recently, in Germany, the Christian Democratic Union (CDU) has started embracing large parts of the Green Party's environmental agenda in order to attract its voters. 

\subsection{Robustly Optimal Bilateral Trade Modifications}

There is a buyer (\(B\)) and a seller (\(S\)). \(S\) possesses an asset that pays out \(v_\theta \in \mathbb{R}\) in state \(\theta\). The state is contractible and the status quo trade agreement sees transfers of size \(\gamma_\theta \in \mathbb{R}\) from \(B\) to \(S\) in each state \(\theta \in \Theta\). We assume that there exist states \(\theta', \theta'' \in \Theta\) for which \(v_{\theta'} > \gamma_{\theta'}\) and \(\gamma_{\theta''} > v_{\theta''}\).

\(B\)'s state-dependent payoff from transacting is \(v_\theta -\gamma_\theta\), and \(S\)'s is \(\gamma_\theta\) (budget balance). \(B\)'s outside option is the sure-thing \(0\) and \(S\)'s outside option is the asset, with random payoff \(v_\theta\) in each state \(\theta \in \Theta\).

We suppose that the status-quo is acceptable--given the agents' subjective beliefs about uncertainty and their (concave) utility functions, each is willing to participate in the arrangement. What modifications to the arrangement are robustly optimal in the sense that \(B\) and \(S\) will still remain willing to participate? Theorem \ref{bigone} reveals the answer.

We observe that for all \(\theta \in \Theta\) and for any \(\lambda \in \left[0,1\right]\),
\[\lambda \left(v_\theta - \gamma_\theta\right) = v_\theta - \lambda \gamma_\theta - \left(1-\lambda\right) v_\theta\text{,}\]
so that any new transfer that yields the requisite convex combination of \(B\)'s status quo payoff from accepting the terms and \(B\)'s outside option of \(0\) equals 
\[\lambda \gamma_\theta + \left(1-\lambda\right) v_\theta\text{,}\]
which is precisely the convex combination (with the same weight) of \(S\)'s status quo payoff and her outside option. This, plus our assumption of budget balance, implies
\begin{remark}
    The buyer and seller will both be willing to participate in a new trade agreement \(\left(\hat{\gamma}_\theta\right)_{\theta \in \Theta}\) if and only if there exists some \(\lambda \in \left(0,1\right]\) such that \(\hat{\gamma}_\theta = \lambda \gamma_\theta + \left(1-\lambda\right) v_\theta\) for all \(\theta \in \Theta\).
\end{remark}

\subsection{Robustly Optimal Insurance Modifications}

Now consider the scenario of a risk-neutral insurer and a consumer. The consumer's payoff without insurance is \(v_\theta\) in each state \(\theta \in \Theta\). The status-quo policy yields a payoff (net of the loss) of \(\alpha_\theta\) to the consumer in each state \(\theta \in \Theta\). We assume that there exist some \(\theta', \theta'' \in \Theta\) for which \(\alpha_{\theta'} > v_{\theta'}\) and \(\alpha_{\theta''} < v_{\theta''}\).
\begin{remark}
    The insurer and consumer will both be willing to accept a new contract \(\left(\hat{\alpha}_\theta\right)_{\theta \in \Theta}\) if and only if there exists some \(\lambda \in \left(0,1\right]\) such that \(\hat{\alpha}_\theta \geq \lambda \alpha_\theta + \left(1-\lambda\right) v_\theta\) for all \(\theta \in \Theta\).
\end{remark}

One interesting implication of this remark is that the set of acceptable contracts could very well expose the consumer to more risk. Indeed, suppose that the status-quo policy is risk free, i.e., yields the consumer the same wealth in every state. The consumer's outside option (not buying insurance) is, of course, riskier. In this case, any actuarially-fair contract that the consumer must accept also exposes her to risk as her wealth is state-dependent.

\section{Variants and Generalizations}\label{xtensions}

In this section, we look beyond pairwise comparisons where we are uncertain about both the DM's belief and her utility. In particular, we study how to make an action better versus multiple alternatives, not just one, in two different environments. First, we assume that the DM's belief is known, shutting down one of the two dimensions of uncertainty. What makes a lottery better? Second, we restore the main specification in which we seek robustness with respect to both beliefs and utilities.

\subsection{How to Make a Lottery Better}\label{lotteries}
We begin by characterizing how to make an action better when we know a DM's belief but do not know her utility.

We assume that the DM has a finite set of actions \(A\) (with \(\left|A\right| = m+1 \geq 2\)). For any fixed \(\mu\), each action \(a \in A\) induces a lottery \(L_a\). For any \(a \in A\), set \(B = A \setminus \left\{a\right\}\) denotes the DM's set of actions other than \(a\). We also let \(\succ\) denote the DM's preference relation over lotteries that corresponds to her utility function. Our definition of a robust improvement in the known-lottery case is, therefore,\footnote{As we noted in the introduction, \cite{chompchomp} conducts a similar exercise to this, asking when \(L_a \succ L_b\) implies \(L_{\hat{a}} \succeq L_{b}\) for either a risk-averse or monotone DM.}
\begin{definition}
    Action \(\hat{a}\) \(B\)-improves upon action \(a\) over \(\left\{\mu\right\}\) and \(\mathcal{U}\) if, for all utilities in the specified class,
    \[L_a \succ L_b \text{ for all } b \in B \ \Rightarrow \ L_{\hat{a}} \succeq L_b \text{ for all } b \in B\text{.}\]
    \end{definition}
    
Concordant with the rest of this paper, our two leading classes of utility functions are those for a risk-averse DM--continuous, strictly increasing, and concave--and a monotone DM--continuous, and strictly increasing. The details of this definition are slightly inconsistent with the other definitions in this paper; e.g., the exact mirror would be 
\[\text{``For all } u \in \mathcal{U}, \ L_a \underset{(\succ)}{\succeq} L_b \text{ for all } b \in B \ \Rightarrow \ L_{\hat{a}} \underset{(\succ)}{\succeq} L_b \text{ for all } b \in B\text{.''}\]
This asymmetry is only for simplicity, allowing for an especially clean result.\footnote{To elaborate, our approach makes use of the revealed-preference results of \cite{fishburn1975separation}. Our definition allows us to use his Corollary 3, in which each revealed preference is strict. Otherwise, we could appeal to \cite{fishburn1975separation}'s result for the ``finite composite case,'' Theorem 3, which would produce a messier result than what we have (Theorem \ref{theoremlottery}).} 

For a fixed \(a \in A\), we enumerate the lotteries corresponding to the other actions in \(A\) (so, the actions in \(B\)) \(L_1, \dots, L_m\). We say that \(\lambda \in \mathbb{R}^m\) is a convex weight if \(0 \leq \lambda_j \leq 1\) for all \(j \in \left\{1,\dots,m\right\}\) and \(\sum_{j=1}^m \lambda_j \leq 1\). Letting \(\trianglerighteq\) denote the dominance relation in the specified class--\textit{viz.}, if the DM is risk-averse, \(L_a \trianglerighteq L_b\) means that lottery \(L_a\) second-order stochastically dominates (SOSD) lottery \(L_b\); and if the DM is monotone \(L_a \trianglerighteq L_b\) means that \(L_a\) first-order stochastically dominates (FOSD) \(L_b\), we have
\begin{theorem}\label{theoremlottery}
    Action \(\hat{a}\) \(B\)-improves upon action \(a\) over \(\left\{\mu\right\}\) and \(\mathcal U\) if and only if for all \(b \in B\) there exists a convex weight \(\lambda\) such that
    \[\sum_{i=1}^{m} \lambda_i L_i + \left(1-\sum_{i=1}^{m} \lambda_i\right) L_{\hat{a}} \trianglerighteq \sum_{i=1}^{m} \lambda_i L_a + \left(1-\sum_{i=1}^{m} \lambda_i\right) L_b\text{.}\]
\end{theorem}
\begin{proof}
    We observe that action \(\hat{a}\) \(B\)-Improving upon action \(a\) is equivalent to the statement that there does not exist a utility function in the specified class such that
    \[L_a \succ L_b \text{ for all } b \in B, \text{ and } L_b \succ L_{\hat{a}} \text{ for some } b \in B\text{.}\]
    By Corollary 3 in \cite{fishburn1975separation}, this holds if and only if for all \(b \in B\) \[\sum_{i=1}^{m} \lambda_i L_i + \left(1-\sum_{i=1}^{m} \lambda_i\right) L_{\hat{a}} \trianglerighteq \sum_{i=1}^{m} \lambda_i L_a + \left(1-\sum_{i=1}^{m} \lambda_i\right) L_b\text{,}\]
    for some convex weight \(\lambda\).\end{proof}
    When there are just two actions (so \(B = \left\{b\right\}\)), we alter the definition of improvement in the obvious way and have
    \begin{corollary}
        Action \(\hat{a}\) \(b\)-Improves upon action \(a\) over \(\left\{\mu\right\}\) and \(\mathcal{U}\) if and only if there exists \(\lambda \in \left[0,1\right]\) such that
    \[\lambda L_b + \left(1-\lambda\right) L_{\hat{a}} \trianglerighteq \lambda L_a + \left(1-\lambda\right) L_b\text{.}\]
    \end{corollary}

    Fixing the DM's belief in the revealed-preference exercise results in a linear problem, which allows us to use tools that we are unable to employ otherwise. Specifically, Corollary 3 in \cite{fishburn1975separation}--and therefore Theorem \ref{theoremlottery}--is implied by a separating hyperplane theorem. This produces our observation that if $\hat{a}$ does not \(b\)-improve on $a$, then we cannot find a combination of $\hat{a}$ and $b$ that dominates a combination of $b$ and $a$. 

\subsection{One Versus Many}\label{many}

We finish by restoring both dimensions of uncertainty. We look for improvements to \(a\) that make it more attractive in comparison to multiple alternatives, irrespective of the DM's belief and utility. Consider the example in which the manager of firm $1$ is deciding between products $a$ and $\hat{a}$, but now the firm is competing against a product from firm $2$ as well as a product from firm $3$. How does the inclusion of firm $3$ affect the manager's decision? 

The main difficulty with considering multiple alternatives in general is that for different levels of risk aversion, the DM's ``next-best'' action may be different. To illustrate, consider the following example with actions $a$, $b$, and $c$ and two states $\theta \in \{0,1\}$. Let $c_0 = c_1 = 2/5$, $a_1 = b_0 = 0$, and $a_0 = b_1 = 1$. Then if the DM is risk-neutral, only $a$ and $b$ will ever be chosen for any belief. If, however, the DM is more risk-averse, say $u(x) = \sqrt{x}$, then for the most uncertain beliefs, $c$ is the best choice. This complicates the analysis of the multi-alternative environment significantly because the set of relevant alternatives depends on the level of risk aversion. 

Two additional assumptions make our analysis of multiple alternatives tractable. We continue to assume that \(A\) is finite, with \(m+1 \geq 2\) actions and maintain our convention that for a given \(a\), \(B = A \setminus \left\{a\right\}\). First, we assume that the state space is \textbf{Rich:} for each action \(a \in A\), there exists at least one state \(\theta \in \Theta\) such that \(a_\theta > \max_{b \in B} b_\theta\)--that is, each action is uniquely optimal in at least one state. Second, we assume that preferences take the following \textbf{Single-Peaked} form: for each action \(a \in A\), and any pair of states \(\theta, \theta' \in \Theta\), either
\begin{enumerate}
    \item there exists \(\lambda \in \left[0,1\right]\) such that \[\lambda a_\theta + (1-\lambda) a_{\theta'} \geq \max_{b \in B}\left\{\lambda b_\theta + (1-\lambda) b_{\theta'}\right\}, \quad \text{or}\]
    \item there exists \(b \in B\) such that \(b_\theta \geq a_\theta\) and \(b_{\theta'} \geq a_{\theta'}\).
\end{enumerate}
The meaning of single-peaked is rather subtle. For any action \(a\) and any pair of states, if there exists a belief supported only on those two states that rationalizes that action, then $a$ must always be included in the set of relevant actions. If instead there exists an action that dominates \(a\) on those two states, then action $a$ need never be included as a relevant action. Either way, this results in a constant set of relevant actions, regardless of the level of risk aversion.  

This set of specifications is satisfied by the oft-encountered ``quadratic loss'' monetary reward: \(\Theta = \left[0,1\right]\), \(A \subset \left[0,1\right]\) and \(a_\theta = -\left(a-\theta\right)^2\). It also holds in our earlier setting of just two actions.

Our superiority definition is
\begin{definition}
 Action \(\hat{a}\) is \textit{\(B\)-superior} to action \(a\) over \(\Delta\) and \(\mathcal{U}\) if for all \(\mu \in \Delta\), for all \(u \in \mathcal{U}\), \[\mathbb{E}_{\mu}u\left(a_{\theta}\right) \underset{(>)}{\geq} \max_{b \in B}\mathbb{E}_{\mu}u\left(b_{\theta}\right) \quad \Rightarrow \quad \mathbb{E}_{\mu} u\left(\hat{a}_{\theta}\right) \underset{(>)}{\geq}\max_{b \in B}\mathbb{E}_{\mu} u\left(b_{\theta}\right)\text{.}\]\end{definition}
We extend our notation in an intuitive way: \[\begin{split}
    \mathcal{A}_B \coloneqq &\left\{\theta \in \Theta \colon a_{\theta} > \max_{b \in B}b_{\theta}\right\}, \quad \mathcal{B}_B \coloneqq \left\{\theta \in \Theta \colon a_{\theta} < \max_{b \in B}b_{\theta}\right\},\\ &\text{and} \quad \mathcal{C}_B \coloneqq \left\{\theta \in \Theta \colon a_{\theta} = \max_{b \in B} b_{\theta}\right\}\text{.}
\end{split}\]
Our result is a natural generalization of Theorem \ref{bigone}.
\begin{proposition}\label{manyactions}
    Fix \(a\) and \(B\). Action \(\hat{a}\) is \(B\)-superior to action \(a\) over \(\Delta\) and \(\mathcal{U}_{RA}\) if and only if i) \(\hat{a}_\theta > \max_{b \in B} b_\theta\) for all \(\theta \in \mathcal{A}_B\); and ii) \(\hat{a}\) (weakly) dominates a mixture of \(a\) and \(b\), for all \(b \in B\).
\end{proposition}
\begin{proof}
    See Appendix \ref{manyactionsproof}.\end{proof}
As we note in the proof, that aggregating Theorem \ref{bigone}'s pairwise conditions implies that \(\hat{a}\) improves upon \(a\) versus a set of actions is relatively simple. That this feature is necessary is more surprising, and this necessity is engendered by the additional richness and single-peakedness structure we place on the decision problem.

\section{Information Acquisition}\label{infoac}

Up until this point we have focused on a simple decision problem in which a DM has a belief over the states of the world and chooses an optimal action based on this belief. It is easy to see that \(b\)-Superiority extends to the scenario in which a DM obtains exogenous information before making a decision. That is, \(\hat{a}\) is \(b\)-Superior to \(a\) if and only if for any \(\mu \in \inter \Delta\) and any Bayes-plausible (martingale) distribution over posteriors \(F\), the DM chooses \(\hat{a}\) with a higher probability than she chooses \(a\) (where, in both cases, \(b\) is chosen with the complementary probability). 

There has been a recent explosion of interest in problems with endogenous information acquisition. Suppose we want to know what the properties of \(\hat{a}\) are that make it chosen more versus \(b\) than \(a\) is when the DM flexibly acquires information before taking a choice. 

We begin with the binary-state environment (\(\Theta = \left\{0,1\right\}\)). We maintain our assumptions that neither \(a\) nor \(b\) dominates the other and that \(a\) is uniquely optimal in state \(0\) and \(b\) in state \(1\). Given a prior \(\mu_{0} \in \left(0,1\right)\), and defining an agent's value function, in belief \(\mu\), as
\[V(\mu) \coloneqq \max_{\tilde{a} \in \left\{a,b\right\}}\mathbb{E}_\mu u(\tilde{a}_\theta)\text{,}\]
the DM's flexible information acquisition problem is
\[\tag{\(1\)}\label{ri1}\max_{F \in \mathcal{F}\left(\mu_{0}\right)}\int_{\Delta}V\left(\mu\right)dF\left(\mu\right) - D\left(F\right)\text{,}\]
where \(\mathcal{F}\left(\mu_{0}\right)\) is the set of Bayes-plausible distributions given prior \(\mu_0\) and \(D\) is a uniformly posterior-separable (UPS) cost functional.\footnote{See, e.g., \cite{caplin2022rationally}. \(D \colon \Delta^2 \to \mathbb{R}\) is UPS if \(D\left(F\right) = \int_{\Delta}c\left(\mu\right)dF\left(\mu\right) - c\left(\mu_{0}\right)\) for some strictly convex and twice continuously differentiable on \(\inter \Delta\) function \(c \colon \Delta \to \mathbb{R}\).} Similarly, the value function when the menu is \(\hat{a}\) and \(b\) is \(\hat{V}(\mu) \coloneqq \max_{\tilde{a} \in \left\{\hat{a},b\right\}}\mathbb{E}_\mu u(\tilde{a}_\theta)\), and when she has this menu the DM solves
\[\tag{\(2\)}\label{ri2}\max_{\hat{F} \in \mathcal{F}\left(\mu_{0}\right)}\int_{\Delta}\hat{V}\left(\mu\right)d\hat{F}\left(\mu\right) - D\left(\hat{F}\right)\text{.}\]

Any solution \(F^*\) (\(\hat{F}^*\)) to Program \ref{ri1} (\ref{ri2}) produces an optimal choice probability of action \(a\) (\(\hat{a}\)) by the DM, which we define to be \(p_{D,\mu_0, u} \equiv p\) (\(\hat{p}_{D,\mu_0, u} \equiv \hat{p}\)).
\begin{definition}
    \emph{\(\hat{a}\) is selected more than \(a\)} if for any UPS \(D\), prior \(\mu_0 \in \inter \Delta\), and utility \(u \in \mathcal{U}_{RA}\); for any optimal \(p\), there exists an optimal choice probability \(\hat{p} \geq p\).
\end{definition}
\begin{proposition}\label{choiceprop}
    \(\hat{a}\) is selected more than \(a\) if and only if \(\hat{a}\) dominates \(a\) or \(b\).
\end{proposition}
\begin{proof}
    A full proof may be found in Appendix \ref{choicepropproof}.
\end{proof}
Most of the work in proving the proposition is in the necessity portion. First consider what happens if $\hat{a}$ is not $b$-superior to $a$. Then, we can find a \(u\) such that \(\hat{p} < \bar{p}\). Take such a \(u\), in which case we can then always find a cost function such that \(a\) is selected with probability \(1\) when the menu is \(\left\{a,b\right\}\) but when the menu is \(\left\{\hat{a},b\right\}\), either the DM learns in a nontrivial manner, and hence selects \(\hat{a}\) with probability strictly less than \(1\), or does not learn but selects \(b\) with probability \(1\). 

To understand why $\hat{a}$ must dominate $a$ or $b$, consider how the DM chooses to learn when there are two states and two actions. The DM chooses an $F^\ast$ such that there are two posteriors, \(\mu_{L}^a\) and $\mu_H^a \geq \mu_L^a$ such that she chooses $a$ with probability 1 if $\mu_0 \leq \mu_L^a$, $b$ with probability 1 if $\mu_0 \geq \mu_H^a$, and $a$ with probability $p \in (0,1)$ if $\mu_0 \in (\mu_L^a, \mu_H^a)$. 

Next, take $\hat{a}$ to be such that $\ell_{\hat{a}}$ is a counterclockwise rotation from $\ell_a$ towards $\ell_b$ (recall Figure \ref{fig1}). This is equivalent, however, to an increase in learning cost (if payoffs had remained the same) for a risk-neutral DM. A higher marginal cost of information leads to less learning; i.e., $(\mu_L^{\hat{a}}, \mu_H^{\hat{a}}) \subseteq (\mu_L^a, \mu_H^a)$. This means that we can find a prior belief such that $\mu_0 \in (\mu_L^a, \mu_H^a)$ but $\mu_0 \geq \mu_H^{\hat{a}}$. In other words, action $a$ has a positive probability of being chosen when the menu is $\{a,b\}$, but action $\hat{a}$ will never be chosen when the menu is $\{\hat{a},b\}$. Hence, $b$-superiority is clearly not enough to guarantee that $\hat{a}$ is chosen more than $a$.

\subsection{Three or More States}

\begin{figure}
        \centering
        \includegraphics[width=.6\linewidth]{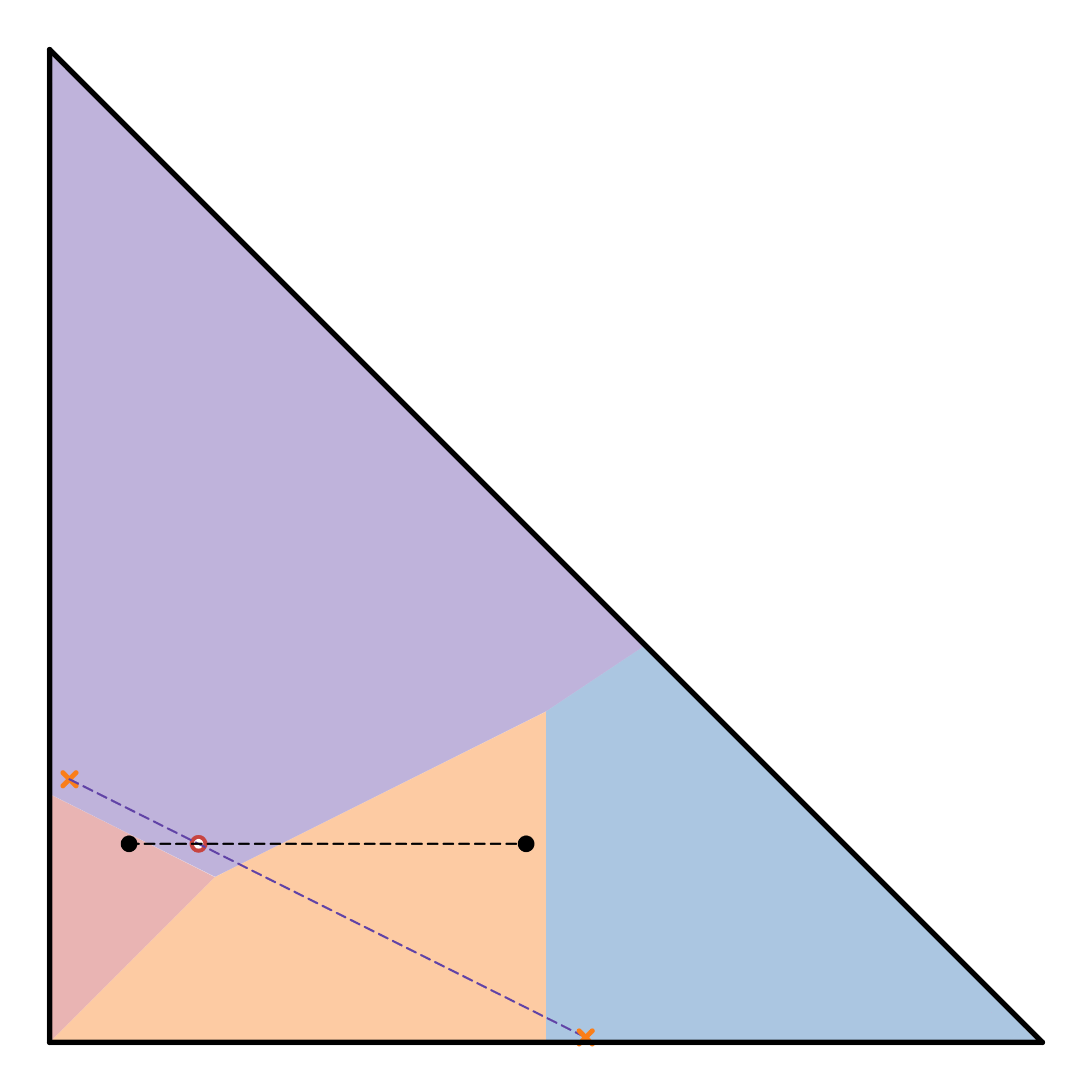}
        \caption{The insufficiency of dominance for three or more states}
        \label{fig2}
    \end{figure}

Alas, when there are three states, \(\hat{a}\) dominating \(a\) no longer implies that \(\hat{a}\) is selected more than \(a\), as we now illustrate. Let \(\Theta = \left\{0,1,2\right\}\) with \(\mu^2 \coloneqq \mathbb{P}(2)\) and \(\mu^1 \coloneqq \mathbb{P}(1)\). We denote \(\alpha_\theta \equiv u(a_\theta) \) for all \(\theta\). Let \(\hat{\alpha}_2 = 3/2 = \alpha_2 + 1/2\), \(\alpha_1 = -1 = \hat{\alpha}_1\) and \(\hat{\alpha}_0 = \alpha_0 = 0\); and \(\beta_\theta = 0\) for all \(\theta \in \Theta\). Evidently, \(\hat{a}\) weakly dominates \(a\) (and \(\hat{a} \neq a\)).

    We have \[V(\mu^1,\mu^2) = \max\left\{\mu^2 - \mu^1,0\right\}, \quad \text{and} \quad \hat{V}(\mu^1,\mu^2) = \max\left\{\frac{3}{2}\mu^2 - \mu^1,0\right\}\text{.}\]
    We then define
    \[W(\mu^1,\mu^2) \coloneqq \max\left\{V, \hat{V} + \frac{1}{4}\mu^1-\frac{1}{8}\right\}\text{,}\]
    and appeal to Lemma A.1 in \cite{flexibilitypaper}, which states that if a set of produced beliefs does not contain multiple posteriors that justify the same action, then there exists a cost function that generates it. We therefore conclude that there exists a UPS cost such that the four points
    \[\tag{\(3\)}\label{ex4}\left\{\left(\frac{2}{25},\frac{1}{5}\right), \left(\frac{12}{25},\frac{1}{5}\right), \left(\frac{1}{50},\frac{53}{200}\right), \left(\frac{27}{50},\frac{1}{200}\right)\right\}\]
    are the support of optimal learning given value function \(W\). Accordingly, we take such a cost function and fix prior \(\left(3/20, 1/5\right)\), then compute \(p = 33/40 > 3/4 = \hat{p}\).\footnote{This calculation simply applies Bayes plausibility to the given posteriors. Specifically, $3/20 = 2/25 p + 12/25 (1 - p)$.}

    Why is $\hat{a}$ chosen less often, even though it dominates $a$? When the DM faces a choice between $a$ and $b$, she only learns about whether the state is $0$ or $1$, as Expression \ref{ex4} reveals. She would like to avoid choosing $a$ in state \(1\), so this is the best approach. If, however, the DM has a choice between $\hat{a}$ and $b$, she now learns about all three, either becoming relatively confident that the state is not \(1\) ($\hat{a}$ is preferable) or relatively confident that the state is not \(2\) ($b$ is preferable). This is enough to make it so that $a$ is chosen more than $\hat{a}$.

    Figure \ref{fig2} illustrates this example on the \(2\)-simplex. The four colored polygons are the regions in which \(V\) lies above the rotated value function \(\hat{V}\) and either \(a\) is optimal (red) or \(b\) is optimal (orange); and the rotated value function \(\hat{V}\) lies above \(V\) and either \(\hat{a}\) is optimal (purple) or \(b\) is optimal (blue). The prior is the hollow red dot. The black dots are the support of \(F^*\), and the orange \(x\)s are the support of \(\hat{F}^*\).

    Leaving the single-dimensional environment corresponding to the two-state simplex engenders this result. What we are doing is taking a value function that in a sense ``twist things,'' making it so that the learning in the two problems is not along the same line segment. There is enough freedom then to pick an appropriate prior so that \(p > \hat{p}\). It is important to note that the relative freedom afforded by the class of UPS cost functions is crucial--if learning must be via a sequential-sampling procedure, if \(\hat{a}\) dominates \(a\) then \(\hat{a}\) is selected more than \(a\) (Proposition \(7\) in \cite{gonccalves2022sequential}).

    We began this paper with the examples of a firm looking to robustly improve its product versus another, or a political party seeking to enlarge its platform. This example reveals that if beliefs are endogenous, even a seemingly obvious improvement can be harmful. That is, a firm's product can be made unambiguously better--even via a price decrease--yet the competitor's product may be chosen more. An indisputable upgrade to a political party's policies can result in a lost election.

\bibliography{sample.bib}

@article{kartik2023single,
  title={Single-Crossing Differences in Convex Environments},
  author={Kartik, Navin and Lee, SangMok and Rappoport, Daniel},
  journal={Review of Economic Studies},
  pages={rdad103},
  year={2023}}

@article{choi2017ordinal,
  title={Ordinal aggregation results via Karlin's variation diminishing property},
  author={Choi, Michael and Smith, Lones},
  journal={Journal of Economic theory},
  volume={168},
  pages={1--11},
  year={2017}}

@article{curello2019preference,
  title={The preference lattice},
  author={Curello, Gregorio and Sinander, Ludvig},
  journal={Mimeo},
  year={2019}
}

@article{milgrom1994monotone,
  title={Monotone comparative statics},
  author={Milgrom, Paul and Shannon, Chris},
  journal={Econometrica},
  pages={157--180},
  year={1994}}

@article{athey2002monotone,
  title={Monotone comparative statics under uncertainty},
  author={Athey, Susan},
  journal={The Quarterly Journal of Economics},
  volume={117},
  number={1},
  pages={187--223},
  year={2002}}

@article{edlin1998strict,
  title={Strict single crossing and the strict Spence-Mirrlees condition: a comment on monotone comparative statics},
  author={Edlin, Aaron S and Shannon, Chris},
  journal={Econometrica},
  volume={66},
  number={6},
  pages={1417--1425},
  year={1998}}

@article{gonccalves2022sequential,
  title={Sequential Sampling Equilibrium},
  author={Gon{\c{c}}alves, Duarte},
  journal={Mimeo},
  year={2022}}

@article{chompchomp,
    author = {Ashvinkumar, Vishal},
    title = {Comparing Comparisons},
    journal = {Mimeo},
    year = {2025}}

@article{flexibilitypaper,
  title={Making Information More Valuable},
  author={Whitmeyer, Mark},
  journal={Mimeo},
  year={2023}
}

@article{markcs,
  title={Calibrating the Subjective},
  author={Whitmeyer, Mark},
  journal={Mimeo},
  year={2024}
}

@article{rabin2000risk,
  title={Risk Aversion and Expected-Utility Theory: A Calibration Theorem},
  author={Rabin, Matthew},
  journal={Econometrica},
  volume={68},
  number={5},
  pages={1281--1292},
  year={2000}}

@article{markscs,
  title={Comparative Statics for the Subjective},
  author={Whitmeyer, Mark},
  journal={Mimeo},
  year={2025}
}

@article{richter1978revelations,
  title={Revelations of a gambler},
  author={Richter, Marcel K and Shapiro, Leonard},
  journal={Journal of Mathematical Economics},
  volume={5},
  number={3},
  pages={229--244},
  year={1978}}

@article{echenique2015savage,
  title={Savage in the Market},
  author={Echenique, Federico and Saito, Kota},
  journal={Econometrica},
  volume={83},
  number={4},
  pages={1467--1495},
  year={2015}}

@article{safety,
Author = {Pease, Marilyn and Whitmeyer, Mark},
Title = {Playing it Safe: Actions Attractive to the Risk Averse},
Journal = {Mimeo},
Year = {2023}}

@article{jewitt1989choosing,
  title={Choosing between risky prospects: the characterization of comparative statics results, and location independent risk},
  author={Jewitt, Ian},
  journal={Management Science},
  volume={35},
  number={1},
  pages={60--70},
  year={1989}}

@article{battigalli2016note,
  title={A note on comparative ambiguity aversion and justifiability},
  author={Battigalli, Pierpaolo and Cerreia-Vioglio, Simone and Maccheroni, Fabio and Marinacci, Massimo},
  journal={Econometrica},
  volume={84},
  number={5},
  pages={1903--1916},
  year={2016}}

@article{weinstein2016effect,
  title={The effect of changes in risk attitude on strategic behavior},
  author={Weinstein, Jonathan},
  journal={Econometrica},
  volume={84},
  number={5},
  pages={1881--1902},
  year={2016}
}

@article{ROTHSCHILD1970225,
title = {Increasing risk: I. A definition},
journal = {Journal of Economic Theory},
volume = {2},
number = {3},
pages = {225-243},
year = {1970},
author = {Michael Rothschild and Joseph E Stiglitz}
}

@article{aumann2008economic,
  title={An economic index of riskiness},
  author={Aumann, Robert J and Serrano, Roberto},
  journal={Journal of Political Economy},
  volume={116},
  number={5},
  pages={810--836},
  year={2008}}

@article{foster2009operational,
  title={An operational measure of riskiness},
  author={Foster, Dean P and Hart, Sergiu},
  journal={Journal of Political Economy},
  volume={117},
  number={5},
  pages={785--814},
  year={2009}}

@article{bali2011generalized,
  title={A generalized measure of riskiness},
  author={Bali, Turan G and Cakici, Nusret and Chabi-Yo, Fousseni},
  journal={Management science},
  volume={57},
  number={8},
  pages={1406--1423},
  year={2011}}

@article{riedel2015foster,
  title={The Foster--Hart measure of riskiness for general gambles},
  author={Riedel, Frank and Hellmann, Tobias},
  journal={Theoretical Economics},
  volume={10},
  number={1},
  pages={1--9},
  year={2015}}

@article{quah2012aggregating,
  title={Aggregating the single crossing property},
  author={Quah, John K-H and Strulovici, Bruno},
  journal={Econometrica},
  volume={80},
  number={5},
  pages={2333--2348},
  year={2012}}

@article{hammond1974simplifying,
  title={Simplifying the choice between uncertain prospects where preference is nonlinear},
  author={Hammond III, John S},
  journal={Management Science},
  volume={20},
  number={7},
  pages={1047--1072},
  year={1974}
}

@article{jewitt1987risk,
  title={Risk aversion and the choice between risky prospects: the preservation of comparative statics results},
  author={Jewitt, Ian},
  journal={The Review of Economic Studies},
  volume={54},
  number={1},
  pages={73--85},
  year={1987}}

@article{karlin1963generalized,
  title={Generalized convex inequalities.},
  author={Karlin, Samuel and Novikoff, Albert},
  year={1963},
    journal={Pacific Journal of Mathematics},
    volume={13},
number={4}}

@article{lambert1979attitudes,
  title={Attitudes to risk},
  author={Lambert, Peter J and Hey, John D},
  journal={Economics Letters},
  volume={2},
  number={3},
  pages={215--218},
  year={1979}}

@article{caplin2022rationally,
  title={Rationally Inattentive Behavior: Characterizing and Generalizing Shannon Entropy},
  author={Caplin, Andrew and Dean, Mark and Leahy, John},
  journal={Journal of Political Economy},
  year={2022},
  pages = {1676--1715},
  number = {6},
  volume = {130},
}

@article{matvejka2015rational,
  title={Rational inattention to discrete choices: A new foundation for the multinomial logit model},
  author={Mat{\v{e}}jka, Filip and McKay, Alisdair},
  journal={American Economic Review},
  volume={105},
  number={1},
  pages={272--298},
  year={2015}}

@MISC {327966,
    TITLE = {Existence of a strictly convex function interpolating given gradients and values},
    AUTHOR = {Dmitri Panov (https://mathoverflow.net/users/943/dmitri-panov)},
    HOWPUBLISHED = {MathOverflow},
    NOTE = {URL:https://mathoverflow.net/q/327966 (version: 2019-04-14)},
    year={2019},
    EPRINT = {https://mathoverflow.net/q/327966},
    URL = {https://mathoverflow.net/q/327966}
}

@article{pomatto2020stochastic,
  title={Stochastic dominance under independent noise},
  author={Pomatto, Luciano and Strack, Philipp and Tamuz, Omer},
  journal={Journal of Political Economy},
  volume={128},
  number={5},
  pages={1877--1900},
  year={2020}}

@article{muller2021rational,
  title={Rational inattention via ignorance equivalence},
  author={Muller-Itten, Michele and Armenter, Roc and Stangebye, Zachary},
  year={2021},
  journal={Mimeo}
}

@article{gilboa2022were,
  title={What were you thinking? Decision theory as coherence test},
  author={Gilboa, Itzhak and Samuelson, Larry},
  journal={Theoretical Economics},
  volume={17},
  number={2},
  pages={507--519},
  year={2022}}

@article{fishburn1975separation,
  title={Separation theorems and expected utilities},
  author={Fishburn, Peter C},
  journal={Journal of Economic Theory},
  volume={11},
  number={1},
  pages={16--34},
  year={1975}}
\bibliographystyle{plainnat}

\appendix

\section{Omitted Proofs}\label{appendix}

\subsection{Theorem \ref{bigone} ``Only If'' Proof}\label{Thm_Main}

We say that \emph{\(\hat{a}\) pairwise-dominates a collection of mixtures of \(a\) and \(b\)} if for any pair \(\left(\theta,\theta'\right) \in \mathcal{A} \times \mathcal{B}\), there exists a \(\lambda_{\theta,\theta'} \in \left[0,1\right]\) such that 
\[\label{in1}\tag{\(1\)}\hat{a}_\theta \geq \lambda_{\theta,\theta'} a_\theta + \left(1-\lambda_{\theta,\theta'}\right)b_\theta \quad \text{and} \quad \hat{a}_{\theta'} \geq \lambda_{\theta,\theta'} a_{\theta'} + \left(1-\lambda_{\theta,\theta'}\right)b_{\theta'}\text{.}\]
In other words, $\hat{a}$ gives a higher payoff than a mixture of $a$ and $b$ in both $\theta$ and $\theta^\prime$.

The first of the two necessity steps is showing that \(b\)-superiority implies pairwise dominance. 
\begin{lemma}\label{lemma:Nec1}
    \(\hat{a}\) is \(b\)-superior to \(a\) only if i) for any \(\theta^{\dagger} \in \mathcal{C}\), \(\hat{a}_{\theta^{\dagger}} \geq a_{\theta^{\dagger}}\); ii) for any \(\theta \in \mathcal{A}\), \(\hat{a}_\theta > b_{\theta}\); and iii) \(\hat{a}\) pairwise-dominates a collection of mixtures of \(a\) and \(b\).
\end{lemma}
\begin{proof}
    The first two conditions are necessary for \(b\)-superiority. We assume they hold and suppose for the sake of contraposition that \(\hat{a}\) does not pairwise-dominate a collection of mixtures of \(a\) and \(b\). This means that there exists at least one pair \(\left(\theta,\theta'\right) \in \mathcal{A} \times \mathcal{B}\) for which there exists no \(\lambda_{\theta, \theta'} \in \left[0,1\right]\) such that the two inequalities in Expression \ref{in1} hold.

    Fix such a pair of \(\theta\) and \(\theta'\). Given utility \(u\), let \(\bar{\mu}_u\) solve
    \[\left(1-\bar{\mu}_u\right) u\left(b_{\theta}\right) + \bar{\mu}_u u\left(b_{\theta'}\right) = \left(1-\bar{\mu}_u\right)u\left(a_{\theta}\right) + \bar{\mu}_u u\left(a_{\theta'}\right)\text{,}\]
    and \(\hat{\mu}_u\) solve
    \[\left(1-\hat{\mu}_u\right) u\left(b_{\theta}\right) + \hat{\mu}_u u\left(b_{\theta'}\right) = \left(1-\hat{\mu}_u\right)u\left(\hat{a}_{\theta}\right) + \hat{\mu}_u u\left(\hat{a}_{\theta'}\right)\text{;}\]
    \textit{viz.,}
    these are the beliefs on the edge of \(\Delta\) between the degenerate probability distributions on states \(\theta\) and \(\theta'\), \(\delta_\theta\) and \(\delta_{\theta'}\), at which the DM with utility \(u\) is indifferent between \(a\) and \(b\) and \(\hat{a}\) and \(b\), respectively. If the DM is risk-neutral, we drop the subscript and simply write \(\bar{\mu}\) and \(\hat{\mu}\).
    
    By assumption \(\bar{\mu}, \hat{\mu} \in \left(0,1\right)\). Moreover, observe that a necessary condition for \(b\)-superiority is \(\hat{\mu} \geq \bar{\mu}\), which holds if and only if 
    \[\label{in2}\tag{\(A2\)} \left(b_{\theta'} - a_{\theta'}\right)\left(\hat{a}_{\theta} - b_{\theta}\right) \geq \left(a_{\theta} - b_{\theta}\right)\left(b_{\theta'} - \hat{a}_{\theta'}\right)\text{.}\]
    If this inequality does not hold, we are done, as we have produced the right negation. Accordingly, suppose it does hold. This means, given our other assumptions,
    \[b_{\theta} < a_{\theta} < \hat{a}_{\theta}, \quad \text{and} \quad \hat{a}_{\theta'} < a_{\theta'} < b_{\theta'}\text{.}\]

    Our goal is to show that \(\bar{\mu}_u > \hat{\mu}_u\) for some strictly-increasing, continuous, and concave \(u\). Let us specify a particular such \(u\). Let
    \[u(x) = \begin{cases}
        x, \quad &\text{if} \quad x \leq a_{\theta'},\\
        \iota x + (1-\iota) a_{\theta'}, \quad &\text{if} \quad x > a_{\theta'}\text{.}
        \end{cases}\]
        for some \(\iota \in (0,1)\), so \(\bar{\mu}_u - \hat{\mu}_u\) equals
    \[\frac{u\left(a_{\theta}\right)-u\left(b_{\theta}\right)}{u\left(a_{\theta}\right) - u\left(b_{\theta}\right) + \iota\left(b_{\theta'}-a_{\theta'}\right)} - \frac{u\left(\hat{a}_{\theta}\right)-u\left(b_{\theta}\right)}{u\left(\hat{a}_{\theta}\right) - u\left(b_{\theta}\right) + \iota\left(b_{\theta'}-a_{\theta'}\right) + a_{\theta'} - \hat{a}_{\theta'}}\text{.}\]
    If \(b_{\theta} \geq a_{\theta'}\), this equals
    \[\frac{a_{\theta} - b_{\theta}}{a_{\theta} -b_{\theta} + b_{\theta'}-a_{\theta'}} - 0 > 0\text{,}\]
    when \(\iota = 0\). If \(b_{\theta} < a_{\theta'}\), this equals
    \[1 - \frac{u\left(\hat{a}_{\theta}\right)-b_{\theta}}{u\left(\hat{a}_{\theta}\right) - b_{\theta} + a_{\theta'} - \hat{a}_{\theta'}} > 0\text{,}\]
    when \(\iota = 0\).\end{proof}

    Finally, we show that pairwise domination implies domination of a mixture, closing the proof of necessity. 
\begin{lemma}\label{pairwiselemma}
    If \(\hat{a}\) pairwise-dominates a collection of mixtures of \(a\) and \(b\), \(\hat{a}_{\theta^{\dagger}} \geq a_{\theta^{\dagger}}\) for all \(\theta^{\dagger} \in \mathcal{C}\), and \(\hat{a}_\theta \geq b_\theta\) for all \(\theta \in \mathcal{A}\), then \(\hat{a}\) dominates a mixture of \(a\) and \(b\).
\end{lemma}
\begin{proof}
    Suppose \(\hat{a}\) pairwise-dominates a collection of mixtures of \(a\) and \(b\), \(\hat{a}_{\theta^{\dagger}} \geq a_{\theta^{\dagger}}\) for all \(\theta^{\dagger} \in \mathcal{C}\), and \(\hat{a}_\theta \geq b_\theta\) for all \(\theta \in \mathcal{A}\). This implies that for any pair \(\left(\theta,\theta'\right) \in \mathcal{A} \times \mathcal{B}\),
    \[\min\left\{1,\frac{\hat{a}_\theta - b_\theta}{a_\theta - b_\theta}\right\} \geq \lambda_{\theta,\theta'} \geq \max\left\{0,\frac{b_{\theta'} - \hat{a}_{\theta'}}{b_{\theta'} - a_{\theta'}}\right\}\text{.}\]
    This implies that there exists a \(\lambda \in \left[0,1\right]\) for which
    \[\inf_{\theta \in \mathcal{A}}\left\{1,\frac{\hat{a}_\theta - b_\theta}{a_\theta - b_\theta}\right\} \geq \lambda \geq \sup_{\theta' \in \mathcal{B}}\left\{0,\frac{b_{\theta'} - \hat{a}_{\theta'}}{b_{\theta'} - a_{\theta'}}\right\}\text{,}\]
    and so \(\hat{a}\) dominates a mixture of \(a\) and \(b\).\end{proof}

    \subsection{Theorem \ref{theorem38} Proof}\label{theorem38proof}
    \begin{proof}
    \(\left(\Leftarrow\right)\) This direction is immediate.

    \smallskip

    \noindent \(\left(\Rightarrow\right)\) By Theorem \ref{bigone}, it is necessary that \(\hat{a}\) dominates a mixture of \(a\) and \(b\). Thus, it suffices to suppose for the sake of contraposition that \(\hat{a}\) dominates a mixture of \(a\) and \(b\) but dominates neither \(a\) nor \(b\). This means that there exists some pair \(\left(\theta,\theta'\right) \in \mathcal{A} \times \mathcal{B}\) for which
    \[b_{\theta} < \hat{a}_{\theta} < a_{\theta}, \quad \text{and} \quad a_{\theta'} < \hat{a}_{\theta'} < b_{\theta'}\text{.}\]
    Our goal is to show that \(\bar{\mu}_u > \hat{\mu}_u\) for some strictly-increasing, continuous \(u\). Take
    \[u(x) = \begin{cases}
        \iota x + (1-\iota) \hat{a}_0, \quad &\text{if} \quad x \leq \hat{a}_{\theta},\\
        x, \quad &\text{if} \quad x > \hat{a}_{\theta}\text{,}
        \end{cases}\]
    for some \(\iota \in (0,1)\). Thus,  \(\bar{\mu}_u - \hat{\mu}_u\) equals
    \[\frac{a_{\theta} - \hat{a}_{\theta} - \iota\left(b_{\theta}-\hat{a}_{\theta}\right)}{a_{\theta} - \hat{a}_{\theta} - \iota\left(b_{\theta}-\hat{a}_{\theta}\right) + u(b_{\theta'})-u(a_{\theta'})} - \frac{\iota\left(\hat{a}_{\theta}-b_{\theta}\right)}{\iota\left(\hat{a}_{\theta}-b_{\theta}\right) + u(b_{\theta'}) - u(\hat{a}_{\theta'})}\text{.}\]
    If \(b_{\theta'} \leq \hat{a}_{\theta}\), this equals 
    \[1 - \frac{\hat{a}_{\theta}-b_{\theta}}{\hat{a}_{\theta}-b_{\theta} + b_{\theta'} - \hat{a}_{\theta'}} > 0\text{,}\]
    when \(\iota = 0\). If \(b_{\theta'} > \hat{a}_{\theta}\), it equals
    \[\frac{a_{\theta} - \hat{a}_{\theta}}{a_{\theta} - \hat{a}_{\theta} + b_{\theta'} -u(a_{\theta'})} - 0 > 0\text{,}\]
    when \(\iota = 0\).
\end{proof}

    \subsection{Proposition \ref{manyactions} Proof}\label{manyactionsproof}
    \begin{proof}
         Fix \(a\) and recall the notation \[\begin{split}
    \mathcal{A}_B \coloneqq &\left\{\theta \in \Theta \colon a_{\theta} > \max_{b \in B}b_{\theta}\right\}, \quad \mathcal{B}_B \coloneqq \left\{\theta \in \Theta \colon a_{\theta} < \max_{b \in B}b_{\theta}\right\},\\ &\text{and} \quad \mathcal{C}_B \coloneqq \left\{\theta \in \Theta \colon a_{\theta} = \max_{b \in B} b_{\theta}\right\}\text{.}
\end{split}\]
We also introduce the following notation: for a fixed \(b\),
\[\begin{split}
    \mathcal{A}_b \coloneqq &\left\{\theta \in \Theta \colon a_{\theta} > b_{\theta}\right\}, \quad \mathcal{B}_b \coloneqq \left\{\theta \in \Theta \colon a_{\theta} < b_{\theta}\right\},\\ &\text{and} \quad \mathcal{C}_b \coloneqq \left\{\theta \in \Theta \colon a_{\theta} = b_{\theta}\right\}\text{.}
\end{split}\]
Evidently, \(\mathcal{A}_B \subseteq \mathcal{A}_b\), \(\mathcal{B}_b \subseteq \mathcal{B}_B\), and \(\mathcal{C}_b \cap \mathcal{A}_B = \emptyset\).

\bigskip

\noindent \(\left(\Leftarrow\right)\) First, we argue that the conditions i) \(\hat{a}_\theta > \max_{b \in B} b_\theta\) for all \(\theta \in \mathcal{A}_B\); and ii) \(\hat{a}\) (weakly) dominates a mixture of \(a\) and \(b\), for all \(b \in B\) jointly imply that for all \(b \in B\), \(\hat{a}_\theta \geq b_\theta\) for all \(\theta \in \mathcal{A}_b\). Suppose for the sake of contradiction that for some \(b \in B\) and some \(\theta \in \mathcal{A}_b\), \(\hat{a}_\theta < b_\theta\). But then \(\hat{a}_\theta < b_\theta < a_\theta\), so \(\hat{a}\) does not dominate a mixture of \(a\) and \(b\). Second, following the proof of Theorem \ref{bigone}, we have that for all \(b \in B\), \(a \succeq b\) \(\Rightarrow\) \(\hat{a} \succeq b\), which implies that \(a \succeq \max_{b \in B}\) \(\Rightarrow\) \(\hat{a} \succeq \max_{b \in B}\). 

It remains to show that \(a \succ \max_{b \in B}\) \(\Rightarrow\) \(\hat{a} \succ \max_{b \in B}\). The region of beliefs on which \(\mathbb{E}_\mu u(a_\theta) \geq \max_{b \in B} \mathbb{E}_\mu u(b_\theta)\) is a polyhedral subset of \(\Delta\), \(P_a\). By the Krein-Milman theorem, any \(\mu \in P_a\) can be written as a convex combination of the extreme points of \(P_a\). Moreover, the only extreme points of \(P_a\) for which \(\mathbb{E}_\mu u(a_\theta) > \max_{b \in B} \mathbb{E}_\mu u(b_\theta)\) are the vertices of the simplex corresponding to degenerate distributions on \(\theta\)s in \(\mathcal{A}_B\). Accordingly, any \(\mu \in P_a\) for which \(\mathbb{E}_\mu u(a_\theta) > \max_{b \in B} \mathbb{E}_\mu u(b_\theta)\) can be written as a convex combination of extreme points of \(P_a\) with strictly positive weight on a set of such vertices. But then, for any such \(\mu\), it must also be that \(\mathbb{E}_\mu u(\hat{a}_\theta) > \max_{b \in B} \mathbb{E}_\mu u(b_\theta)\), as \(\hat{a} > \max_{b \in B} b_\theta\) for all \(\theta \in \mathcal{A}_B\) plus the already-shown fact that \(\mathbb{E}_\mu u(\hat{a}_\theta) \geq \max_{b \in B} \mathbb{E}_\mu u(b_\theta)\) for all \(\mu \in P_a\).

\smallskip

\noindent \(\left(\Rightarrow\right)\) We need to show that action \(\hat{a}\) is \(B\)-superior to action \(a\) only if i) \(\hat{a}_\theta > \max_{b \in B} b_\theta\) for all \(\theta \in \mathcal{A}_B\); and ii) for all \(b \in B\) \(\hat{a}\) (weakly) dominates a mixture of \(a\) and \(b\).
         
         First, \(\hat{a} > \max_{b \in B} b_\theta\) for all \(\theta \in \mathcal{A}_B\) is clearly necessary, as otherwise we could just take the degenerate distribution on an offending \(\theta \in \mathcal{A}_B\). Now let this hold but suppose for the sake of contraposition that there exists some \(b \in B\) such that \(\hat{a}\) does not dominate a mixture of \(a\) and \(b\).

         For any pair \(\left(\theta,\theta'\right) \in \mathcal{A}_B \times \mathcal{B}_B\), we say that action \(b^{*} \in B\) is the boundary action if there exists \(\bar{\gamma} \in \left(0,1\right)\) such that for all \(\gamma < \bar{\gamma}\), \[(1-\gamma) a_\theta + \gamma a_{\theta'} > \max_{b \in B} \left\{(1-\gamma) b_\theta + \gamma b_{\theta'}\right\}\]  and \[(1-\bar{\gamma}) a_\theta + \bar{\gamma} a_{\theta'} = (1-\bar{\gamma}) b^{*}_\theta + \bar{\gamma} b^{*}_{\theta'}\text{.}\]
         \begin{claim}
             It is without loss of generality to assume that \(b^{*}\) is unique. Moreover, for any concave and strictly increasing \(u\) and all \(\gamma \in \left[0,1\right]\), \[(1-\gamma) u(a_\theta) + \gamma u(a_{\theta'}) \underset{(>)}{\geq} (1-\gamma) u(b^{*}_\theta) + \gamma u(b^{*}_{\theta'}) \ \Rightarrow \ (1-\gamma) u(a_\theta) + \gamma u(a_{\theta'}) \underset{(>)}{\geq} (1-\gamma) u(b_\theta) + \gamma u(b_{\theta'})\text{,}\] for all \(b \in B\).
         \end{claim}
         \begin{proof}
            Our assumption of single-peakedness means that on any \(\left(\theta,\theta'\right) \in \mathcal{A}_B \times \mathcal{B}_B\), the only dominated actions for a risk-neutral agent are those that are dominated (on these two states) by some other action. Accordingly, the set of undominated actions on this pair of states does not grow in the DM's risk aversion. On the other hand, the two Proposition 1s in \cite{battigalli2016note} and \cite{weinstein2016effect} reveal that increased risk aversion can only increase (in a set-inclusion sense) the rationalizable set. Thus, the set of undominated actions on this pair of states is the same regardless of the DM's concave \(u\). Finally, we set \(b^{*}\) to be the action with the largest value in state \(\theta\) of the undominated actions (on these two states) in \(B\), making an arbitrary selection if multiple such actions have the same payoffs in \(\theta\) and \(\theta'\).
         \end{proof}
        This claim highlights the function of our assumption of single-peakedness. It ensures that regardless of the DM's risk-averse utility, for each pair \(\left(\theta,\theta'\right) \in \mathcal{A}_B \times \mathcal{B}_B\), there is only one other action that matters and it is the same regardless of \(u\).
         
        By Lemma \ref{lemma:Nec1}, it suffices to show that there exists some pair \(\left(\theta,\theta'\right) \in \mathcal{A}_B \times \mathcal{B}_B\) for which there does not exist a \(\lambda_{\theta,\theta'} \in \left[0,1\right]\) such that
         \[\hat{a}_\theta \geq \lambda_{\theta,\theta'} a_\theta + \left(1-\lambda_{\theta,\theta'}\right)b^{*}_\theta \quad \text{and} \quad \hat{a}_{\theta'} \geq \lambda_{\theta,\theta'} a_{\theta'} + \left(1-\lambda_{\theta,\theta'}\right)b^{*}_{\theta'}\text{,}\]
         where \(b^{*}\) is the boundary action. We call this a boundary violation.
        \begin{claim}
            If \(\hat{a}_{\theta'} < a_{\theta'}\) for some \(\theta' \in \mathcal{B}_B\), there exists a boundary violation.
        \end{claim}
        \begin{proof}
            This is an immediate consequence of the fact that \(b_{\theta'}^* > a_{\theta'} > \hat{a}_{\theta'}\), so no averaging of the first two can produce something greater than the third.
        \end{proof}
        Henceforth, we assume that for all \(\theta' \in \mathcal{B}_B\), \(\hat{a}_{\theta'} \geq a_{\theta'}\).
         \begin{claim}\label{2ndclaim}
             If there exists some \(b \in B\) and \(\theta^{\dagger} \in \mathcal{C}_b\) such that \(\hat{a}_{\theta^{\dagger}} < a_{\theta^{\dagger}}\) then \(\hat{a}\) is not \(B\)-superior to \(a\).
         \end{claim}
         \begin{proof}
             Suppose there exists some \(b \in B\) and \(\theta^{\dagger} \in \mathcal{C}_b\) such that \(\hat{a}_{\theta^{\dagger}} < a_{\theta^{\dagger}}\). If \(\hat{a}_\theta < a_\theta\) for some \(\theta \in \mathcal{C}_B\), we have the desired conclusion. If \(\hat{a}_\theta \geq a_\theta\) for all \(\theta \in \mathcal{C}_B\) but there exists some \(b \in B\) and \(\theta \in \mathcal{C}_b\) such that \(\hat{a}_\theta < a_\theta\), it must be the case that \(\hat{a}_{\theta'} < a_{\theta'}\) for some \(\theta' \in \mathcal{B}_B\), a contradiction.
         \end{proof}

         As a result of Claim \ref{2ndclaim}, we assume that for all \(\theta \in \mathcal{C}_b\), \(\hat{a}_\theta \geq a_\theta\).
         \begin{claim}\label{3rdclaim}
             If there exists some \(b \in B\) and \(\theta \in \mathcal{A}_b\) such that \(\hat{a}_\theta < b_\theta\), then \(\hat{a}\) is not \(B\)-superior to \(a\).
         \end{claim}
         \begin{proof}
        Suppose there exists some \(b \in B\) and \(\theta \in \mathcal{A}_b\) such that \(\hat{a}_\theta < b_\theta\). But then either \(\theta \in \mathcal{A}_B\), \(\theta \in \mathcal{C}_B\), or \(\theta \in \mathcal{B}_B\) and \(\hat{a}_\theta < b_\theta < a_\theta\); all contradictions.\end{proof}

        As a result of Claim \ref{3rdclaim}, we assume that for all \(\theta \in \mathcal{A}_b\), \(\hat{a}_\theta \geq b_\theta\).
         
        By Lemma \ref{pairwiselemma}, \(\hat{a}\) does not pairwise dominate a collection of mixtures of \(a\) and \(b\). That is, on some \(\left(\theta,\theta'\right) \in \mathcal{A}_b \times \mathcal{B}_b\) there does not exist a \(\lambda_{\theta,\theta'} \in \left[0,1\right]\) such that
         \[\hat{a}_\theta \geq \lambda_{\theta,\theta'} a_\theta + \left(1-\lambda_{\theta,\theta'}\right)b_\theta \quad \text{and} \quad \hat{a}_{\theta'} \geq \lambda_{\theta,\theta'} a_{\theta'} + \left(1-\lambda_{\theta,\theta'}\right)b_{\theta'}\text{.}\]
         As \(\hat{a}_{\theta'} \geq a_{\theta'}\) for all \(\theta' \in \mathcal{B}_B\) this implies that on some \(\left(\theta,\theta'\right) \in \mathcal{A}_B \times \mathcal{B}_B\), there does not exist a \(\lambda_{\theta,\theta'} \in \left[0,1\right]\) such that
         \[\hat{a}_\theta \geq \lambda_{\theta,\theta'} a_\theta + \left(1-\lambda_{\theta,\theta'}\right)b_\theta \quad \text{and} \quad \hat{a}_{\theta'} \geq \lambda_{\theta,\theta'} a_{\theta'} + \left(1-\lambda_{\theta,\theta'}\right)b_{\theta'}\text{.}\]

         We need to show that such a convex weight does not exist on this pair of states with respect to the boundary action \(b^{*}\). But this is obvious: given \(\hat{a}_{\theta'} \geq a_{\theta'}\), the non-existence of the weight \(\lambda_{\theta,\theta'}\) is equivalent to the line (in \(\gamma\))
         \[(1-\gamma) \hat{a}_\theta + \gamma\hat{a}_{\theta'}\]
         lying strictly below the line
         \[(1-\gamma) a_\theta + \gamma a_{\theta'}\text{,}\]
         for all 
         \[\gamma \leq \frac{a_\theta - b_\theta}{a_\theta - b_\theta + b_{\theta'}-a_{\theta'}}\text{,}\]
         which is the weight equalizing \(a\) and \(b\) on these two states. Moreover, by the definition of the boundary action, the weight equalizing \(a\) and \(b\) on these two states
         \[\frac{a_\theta - b^{*}_\theta}{a_\theta - b^{*}_\theta + b^{*}_{\theta'}-a_{\theta'}} \leq \frac{a_\theta - b_\theta}{a_\theta - b_\theta + b_{\theta'}-a_{\theta'}}\text{,}\]
         implying that there does not exist \(\lambda_{\theta,\theta'} \in \left[0,1\right]\) such that
         \[\hat{a}_\theta \geq \lambda_{\theta,\theta'} a_\theta + \left(1-\lambda_{\theta,\theta'}\right)b^{*}_\theta \quad \text{and} \quad \hat{a}_{\theta'} \geq \lambda_{\theta,\theta'} a_{\theta'} + \left(1-\lambda_{\theta,\theta'}\right)b^{*}_{\theta'}\text{.}\]
         We have produced the desired violation, so we are done.\end{proof}

\subsection{Proposition \ref{choiceprop} Proof}\label{choicepropproof}
\begin{proof}
    \(\left(\Leftarrow\right)\) We want to show that if $\hat{a}$ dominates $a$ or $b$, then $\hat{a}$ is selected more than $a$. If \(\hat{a}\) dominates \(b\), \(\hat{p}=1\) is a solution. Now let \(\hat{a}\) not dominate \(b\) but dominate \(a\). We denote \(\alpha_\theta \equiv u(a_\theta) \leq u(\hat{a}_\theta) \equiv \hat{\alpha}_\theta\) for all \(\theta \in \left\{0,1\right\}\). Without loss of generality we normalize \(u(b_\theta) = 0\) for all \(\theta \in \left\{0,1\right\}\). 
    
    Suppose first that \(\mu_0\) and \(c\) are such that \(p = 1\). This is equivalent to \(\mu_0 \leq \bar{\mu}_u\) and 
    \[\left(\alpha_1 - \alpha_0 - c'(\mu_0)\right)\mu + c'(\mu_0) \mu_0 + \alpha_0 \geq 0 \text{,}\]
    for all \(\mu \in \left[0,1\right]\). For any \(\mu \in \left[0,1\right]\) this expression is strictly increasing in both \(\alpha_0\) and \(\alpha_1\), so we have 
    \[\left(\hat{\alpha}_1 - \hat{\alpha}_0 - c'(\mu_0)\right)\mu + c'(\mu_0) \mu_0 + \hat{\alpha}_0 \geq 0 \text{,}\]
    for all \(\mu \in \left[0,1\right]\), i.e., \(\hat{p} = 1\).

    Now suppose that \(\mu_0\) is such that \(p \in \left(0,1\right)\). In this case, the support of \(F^*\), \(\mu_L\) and \(\mu_{H}\), solves
    \[\alpha_1 - \alpha_0 - c'\left(\mu_{L}\right) + c'\left(\mu_{H}\right) = 0\text{,}\]
    and
    \[c'\left(\mu_{L}\right) \mu_L - c'\left(\mu_{H}\right) \mu_{H}  + c\left(\mu_{H}\right) + \alpha_0 - c\left(\mu_{L}\right) = 0\text{.}\]
    Appealing to the implicit function theorem, we obtain
    \[\mu_{L}'(\alpha_1) = \frac{\mu_{H}}{c''\left(\mu_{L}\right) \left(\mu_{H} - \mu_{L}\right)} > 0, \quad \text{and} \quad \mu_{H}'(\alpha_1) = \frac{\mu_{L}}{c''\left(\mu_{H}\right) \left(\mu_{H} - \mu_{L}\right)} > 0\text{.}\]
    Likewise,
    \[\mu_{L}'(\alpha_0) = \frac{1-\mu_{H}}{c''\left(\mu_{L}\right) \left(\mu_{H} - \mu_{L}\right)} > 0, \quad \text{and} \quad \mu_{H}'(\alpha_0) = \frac{1-\mu_{L}}{c''\left(\mu_{H}\right) \left(\mu_{H} - \mu_{L}\right)} > 0\text{.}\]
    Evidently, \(p\) is strictly increasing in both \(\mu_{H}\) and \(\mu_{L}\), so \(\hat{p} \geq p\).
    
    \smallskip

    \noindent \(\left(\Rightarrow\right)\) Suppose for the sake of contraposition that \(\hat{a}\) dominates neither \(a\) nor \(b\). If \(\hat{a}\) is dominated (and does not dominate) by \(a\) or \(b\), the outcome is trivial. Accordingly, suppose \(\hat{a}\) is dominated by neither. There are three possibilities: either i) \(\hat{a}_0 > a_0 > b_0\) and \(\hat{a}_1 < a_1 < b_1\); or ii) \(a_0 > \hat{a}_0 > b_0\) and \(b_1 > \hat{a}_1 > a_1\); or iii) \(a_0 > b_0 > \hat{a}_0\) and \(a_1 < b_1 < \hat{a}_1\).

    Case iii is immediate. Let the DM be risk neutral: then there exists a \(\mu' \in \left(0,1\right)\), such that for any belief \(\mu < \mu'\), it is uniquely optimal for the DM to take action \(a\) when her menu is \(\left\{a,b\right\}\) and \(b\) when her menu is \(\left\{\hat{a},b\right\}\). Then, one need only pick a sufficiently convex \(c\)--that such a convex \(c\) can always be found is an implication of Lemma A.1 in \cite{flexibilitypaper}\footnote{This is not strictly true as \cite{flexibilitypaper} does not impose that the cost function is twice continuously differentiable, merely strictly convex. However, it is easy to extend that result to smooth functions: see, e.g. \cite{327966}. Alternatively, one could remove the twice-continuously differentiable specification at the expense of not being able to appeal to the implicit function theorem in the sufficiency portion of the proof.}--such that for prior \(\mu_0\), no learning is uniquely optimal in both problems, in which case \(p = 1 > 0 = \hat{p}\).
    
    Case i is also easy. Observe that \(\hat{a}\) is not \(b\)-superior to \(a\), so there exists a strictly increasing concave \(u\) for which \(0 < \hat{\mu}_u < \bar{\mu}_u < 1\). We maintain the convention \(\alpha_\theta \equiv u(a_\theta)\) and \(\hat{\alpha}_\theta \equiv u(\hat{a}_\theta)\) for all \(\theta \in \left\{0,1\right\}\) and also introduce the notation \(\beta_\theta \equiv u(b_\theta)\) for all \(\theta \in \left\{0,1\right\}\). 
    
    We tweak the notation
    \[\ell_a = \mu \alpha_1 + \left(1-\mu\right) \alpha_0, \ \ell_{\hat{a}} = \mu \hat{\alpha}_1 + \left(1-\mu\right) \hat{\alpha}_0, \text{ and } \ell_b = \mu \beta_1 + \left(1-\mu\right) \beta_0\text{,}\]
    and define \(W(\mu) \coloneqq \max\left\{\ell_a,\ell_{\hat{a}},\ell_b\right\}\). We let \(\tilde{\mu}\) denote the intersection of \(\ell_a\) and \(\ell_{\hat{a}}\), and observe that \(0 < \tilde{\mu} < \bar{\mu}\); this holds because \(\ell_{\hat{a}}\) has a steeper slope and a strictly larger \(y\)-intercept than \(\ell_a\). Then, Lemma A.1 in \cite{flexibilitypaper} implies that for any triple \(\mu_1 \in \left(0,\tilde{\mu}\right)\), \(\mu_2 \in \left(\tilde{\mu}, \bar{\mu}\right)\), and \(\mu_3 \in \left(\bar{\mu},1\right)\), there exists a UPS cost such that when the DM's value function is \(W\), any optimal learning has support on the three specified points. Accordingly, for such a cost function, when \(\mu_0 = \mu_2\), \(p = 1\) and \(\hat{p} < 1\).

    Finally, case ii: the argument from the previous paragraph allows us to assume that \(\hat{a}\) is \(b\)-superior to \(a\), or else we are done. Consequently, \(\bar{\mu}_u \leq \hat{\mu}_u\) for all permissible \(u\). Fix such a \(u\) and normalize payoffs so that \(\beta_\theta = 0\) for all \(\theta \in \left\{0,1\right\}\) (this is without loss of generality, as \(u\) has been fixed). Now take a line \(f(\mu) \coloneqq -\gamma \mu + \delta\), where \(0 < \delta < \alpha_0 - \hat{\alpha}_0\), \(\gamma > \delta\), and \[\frac{\delta}{\gamma} > \hat{\mu}_u = \frac{\hat{\alpha}_0}{\hat{\alpha}_0 - \hat{\alpha}_1}\text{.}\]
    
    Defining 
    \(T(\mu) \coloneqq \max\left\{\ell_a, \max\left\{\ell_{\hat{a}}, 0\right\} + f, 0\right\}\),
    we note that this piecewise-affine curve has three kink points. First, at some \(\mu_1 \in \left(0,1\right)\), where \(\ell_a\) and \(\ell_{\hat{a}} + f\) intersect. Second, at some \(\mu_2 \in \left(\mu_1,1\right)\), where \(0\) and \(\ell_{\hat{a}}\) intersect. Third, at some \(\mu_3 \in \left(\mu_2, 1\right)\) where \(f\) and \(0\) intersect. Again appealing to Lemma A.1 in \cite{flexibilitypaper}, taking a prior \(\mu_0 \in \left(\mu_2, \mu_3\right)\) we note the existence of a UPS cost producing \(p > 0\) and \(\hat{p} = 0\).\end{proof}

\end{document}